\documentclass{article}

\usepackage{algorithm}
\usepackage[noend]{algpseudocode}
\usepackage{amsmath,amssymb,amsthm}
\usepackage{arxiv}
\usepackage[labelfont=bf]{caption}
\usepackage{cite}
\usepackage{float}
\usepackage{graphicx}
\usepackage{hyperref}
\usepackage{mathtools}
\usepackage{makecell}
\usepackage{sidecap} \sidecaptionvpos{figure}{c}
\usepackage{xcolor}

\newif\ifcomment
\commenttrue    

\newif\iffigabbrv
\figabbrvfalse   
\newcommand{\figtext}{\iffigabbrv Fig.\else Figure\fi}

\newtheorem{theorem}{Theorem}
\newtheorem{lemma}[theorem]{Lemma}
\newtheorem{definition}[theorem]{Definition}

\title{Programming Active Cohesive Granular Matter with Mechanically Induced Phase Changes}
\author{
    Shengkai Li \\
        School of Physics\\
        Georgia Institute of Technology\\
        Atlanta, GA 30332, USA \\
        \texttt{shengkaili@gatech.edu} \\
    \And
    Bahnisikha Dutta \\
        School of Electrical\\ and Computer Engineering\\
        Georgia Institute of Technology\\
        Atlanta, GA 30332, USA \\
        \texttt{bahnisikhadutta@gatech.edu} \\
    \And
    Sarah Cannon \\
        Mathematical Sciences\\
        Claremont McKenna College\\
        Claremont, CA 91711, USA \\
        \texttt{scannon@cmc.edu} \\
    \And
    Joshua J.\ Daymude \\
        Computer Science, CIDSE\\
        Arizona State University\\
        Tempe, AZ 85281, USA \\
        \texttt{jdaymude@asu.edu} \\
    \And
    Ram Avinery \\
        School of Physics\\
        Georgia Institute of Technology\\
        Atlanta, GA 30332, USA \\
        \texttt{ravinery3@gatech.edu} \\
    \And
    Enes Aydin \\
        School of Physics\\
        Georgia Institute of Technology\\
        Atlanta, GA 30332, USA \\
        \texttt{enes.aydin@physics.gatech.edu} \\
    \And
    Andr\'ea W.\ Richa \\
        Computer Science, CIDSE\\
        Arizona State University\\
        Tempe, AZ 85281, USA \\
        \texttt{aricha@asu.edu} \\
    \And
    Daniel I.\ Goldman \\
        School of Physics\\
        Georgia Institute of Technology\\
        Atlanta, GA 30332, USA \\
        \texttt{daniel.goldman@physics.gatech.edu} \\
    \And
    Dana Randall \\
        School of Computer Science\\
        Georgia Institute of Technology\\
        Atlanta, GA 30332, USA \\
        \texttt{randall@cc.gatech.edu}
}
\date{}

\begin{document}

\maketitle

\begin{abstract}
    Active matter physics and swarm robotics have provided powerful tools for the study and control of ensembles driven by internal sources.
    At the macroscale, controlling swarms typically utilizes significant memory, processing power, and coordination unavailable at the microscale --- e.g., for colloidal robots, which could be useful for fighting disease, fabricating intelligent textiles, and designing nanocomputers.
    To develop principles that that can leverage physics of interactions and thus can be utilized across scales, we take a two-pronged approach: a theoretical abstraction of self-organizing particle systems and an experimental robot system of active cohesive granular matter that intentionally lacks digital electronic computation and communication, using minimal (or no) sensing and control, to test theoretical predictions.
    We consider the problems of aggregation, dispersion, and collective transport.
    As predicted by the theory, as a parameter representing interparticle attraction increases, the robots transition from a dispersed phase to an aggregated one, forming a dense, compact collective.
    When aggregated, the collective can transport non-robot ``impurities'' in their environment, thus performing an emergent task driven by the physics underlying the transition.
    These results point to a fruitful interplay between algorithm design and active matter robophysics that can result in new nonequilibrium physics and principles for programming collectives without the need for complex algorithms or capabilities.
\end{abstract}

\section*{Introduction}

Self-organizing collective behaviors are found throughout nature, including shoals of fish aggregating to intimidate predators~\cite{Magurran1990}, fire ants forming rafts to survive floods~\cite{Mlot2011}, and bacteria forming biofilms to share nutrients when they are metabolically stressed~\cite{Liu2015}.
Inspired by such systems, researchers in swarm robotics and programmable active matter have used many approaches towards enabling ensembles of simple, independent units to cooperatively accomplish complex tasks~\cite{Brambilla2013,Bayindir2016,Dorigo2020-futureswarms}.
Both control theoretic and distributed computing approaches have achieved some success, but often rely critically on robots computing and communicating complex state information, requiring relatively sophisticated hardware that can be prohibitive at small scales~\cite{Elamvazhuthi2019-meanfield,Flocchini2019}.
Alternatively, statistical physics approaches model swarms as systems being driven away from thermal equilibrium by robot interactions and movements (see, e.g.,~\cite{Mayya2019,Notomista2019}). 
Tools from statistical physics such as the Langevin and Fokker-Planck equations can then be used to analyze the mesoscopic and macroscopic system behaviors~\cite{Hamann2018-swarmrobotics}.
Current approaches present inherent tradeoffs, especially as individual robots become smaller and have limited functional capabilities~\cite{Hines2017-softactuators,Xie2019} or approach the thermodynamic limits of computing and power~\cite{Wolpert2019}.

To apply to a general class of micro- or nano-scale devices with limited capabilities, we focus on systems of autonomous, self-actuated entities that utilize strictly local interactions to induce macroscale behaviors.
Two behaviors of interest are \textit{dynamic free aggregation}, where agents gather together without preference for a specific aggregation site (see Section 3.2.1 of~\cite{Bayindir2016}), and \textit{dispersion}, its inverse.
These problems are widely studied, but most work either considers robots or models with relatively powerful capabilities --- e.g., persistent memory for complex state information~\cite{Rubenstein2014,Piranda2018} or long-range communication and sensing~\cite{Fates2011,Gauci2014,Ozdemir2019} --- or lack rigorous mathematical foundations explaining the generality and limitations of their results as sizes scale~\cite{Garnier2009,Correll2011,Li2019}.
Recent studies on active interacting particles~\cite{Agrawal2017-tunablestructures} and inertial, self-organizing robots~\cite{Deblais2018-boundarycontrol} employ physical models to treat aggregation and clustering behaviors, but neither prove behavior guarantees that scale with system size and volume.
Supersmarticle ensembles~\cite{Savoie2019-smarticleensemble} are significantly more complex, exhibiting many transient behavioral patterns stemming from their many degrees of freedom and chaotic interactions, making them less amenable to rigorous algorithmic analysis.

Here we take a two-pronged approach to understanding the fundamental principles of programming task-oriented matter that can be implemented across scales without requiring sophisticated hardware or traditional computation that leverages the physics of local interactions.
We use a theoretical abstraction of self-organizing particle systems (SOPS), where we can design and rigorously analyze simple distributed algorithms to accomplish specific goals that are flexible and robust to errors.
We then build a new system of deliberately rudimentary active ``cohesive granular robots'' (which, to honor granular physics pioneer Robert Behringer, we call ``BOBbots'' for Behaving, Organizing, Buzzing robots) to test whether the theoretical predictions can be realized in a real-world damped driven system.
Remarkably, the lattice based equilibrium model quantitatively captures the aggregation dynamics of the robots.  
With a provable algorithmic model and even simpler BOBbots capturing the algorithm's essential rules, we next explore how contact stress sensing --- a capability that is readily available in the robotic platform but interestingly not easily computable by a strictly local, distributed algorithm --- can enhance aggregation performance, as suggested by insights from the theoretical model.
This complementary approach demonstrates a new integration of the fields of distributed algorithms, active matter, and granular physics that navigates a translation from theoretical abstraction to practice, utilizing methodologies inherent to each field.

\section*{Results}

\subsection*{Aggregation algorithm}

While many systems use interparticle attraction and sterical exclusion to achieve system-wide aggregation and interparticle repulsion to achieve dispersion, these methods typically use some long-range sensing and tend to be nonrigorous, lacking formal proofs guaranteeing desirable system behavior.
To better understand these collective behaviors, the abstract model of \textit{self-organizing particle systems} (SOPS) allows us to define a formal distributed algorithm and rigorously quantify long-term behavior.
Particles in a SOPS exist on the nodes (or vertices) of a lattice,
with at most one particle per node, and move between nodes along lattice edges.
Each particle is anonymous (unlabeled), interacts only with particles occupying adjacent lattice nodes, and does not have access to any global information such as a coordinate system or the total number of particles.

In earlier work, Cannon et al.~\cite{Cannon2016} analyzed a distributed SOPS algorithm for aggregation and dispersion
under the assumption that the particle system remained simply connected (i.e., the system forms a single connected cluster with no holes).
This SOPS algorithm defines a finite Markov chain with local moves that connect the state space of all simply connected configurations of particles.
Moves are defined so that each particle, when activated by its own Poisson clock (i.e., after a delay chosen at random from a Poisson distribution with constant mean), chooses a random neighboring node and moves there with a probability that is a function of the number of neighbors in the current and new positions provided the node is unoccupied and the move satisfies local conditions that guarantee the configuration stays simply connected.
In particular, for configurations $\sigma$ and $\tau$ differing by the move of a single particle $p$ along a lattice edge, the transition probability is defined as $P(\sigma, \tau) \propto \min(1, \lambda^{n'-n}),$ where $\lambda > 0$ is a bias parameter that is an input to the algorithm, $n$ is the number of neighbors of $p$ in $\sigma$ and $n'$ is the number of neighbors of $p$ in $\tau$.
These probabilities arise from the celebrated Metropolis--Hastings algorithm~\cite{Metropolis1953,Hastings1970} and are defined so that the Markov chain converges to a unique Boltzmann distribution $\pi$ such that $\pi(\sigma)$ is proportional to $\lambda^{E(\sigma)},$ where $E(\sigma)$ is the number of nearest neighbor pairs in $\sigma$ (i.e., those pairs that are adjacent on the lattice).

It was shown in \cite{Cannon2016} that the connected SOPS ensemble provably aggregates into a compact conformation when $\lambda > 3.42$ and expands to a conformation with nearly maximal (linear) perimeter when $\lambda < 2.17$ with high probability, i.e., with a probability of failure that is exponentially small in $N$, the number of particles.
However, despite rigorously achieving both aggregation and dispersion, this distributed algorithm has two notable drawbacks that make it infeasible for direct implementation in a physical system of simple robots: the connectivity requirement that tethers the particles together and the ``look ahead'' requirement used to calculate transition probabilities ensuring convergence to the desired Boltzmann distribution.

To address these issues, we define a modified aggregation and dispersion algorithm $\cal{M}_{\rm AGG}$ where particles can disconnect and moves rely only on the current state.
Here, particles occupy nodes of a finite region of the triangular lattice, again moving stochastically and favoring configurations with more pairs of neighboring particles.
Each particle has its own Poisson clock and, when activated, chooses a random adjacent lattice node.
If that node is unoccupied, the particle moves there with probability $\lambda^{-n}$, where $n$ is the number of current neighbors of the particle, for bias parameter $\lambda > 0$.
Thus, rather than biasing particles towards nodes with more neighbors, we instead discourage moves away from nodes with more neighbors, with larger $\lambda$ corresponding to a stronger ferromagnetic attraction between particles (\figtext~\ref{fig:SOPS}A).
This new chain $\mathcal{M}_\text{AGG}$ converges to the same Boltzmann distribution $\pi(\sigma) \propto \lambda^{E(\sigma)}$ over particle system configurations $\sigma$ as the original SOPS algorithm.
Details of the proofs can be found in the Materials and Methods.

\begin{SCfigure}[][th]
    \centering
    \includegraphics[width=0.65\textwidth]{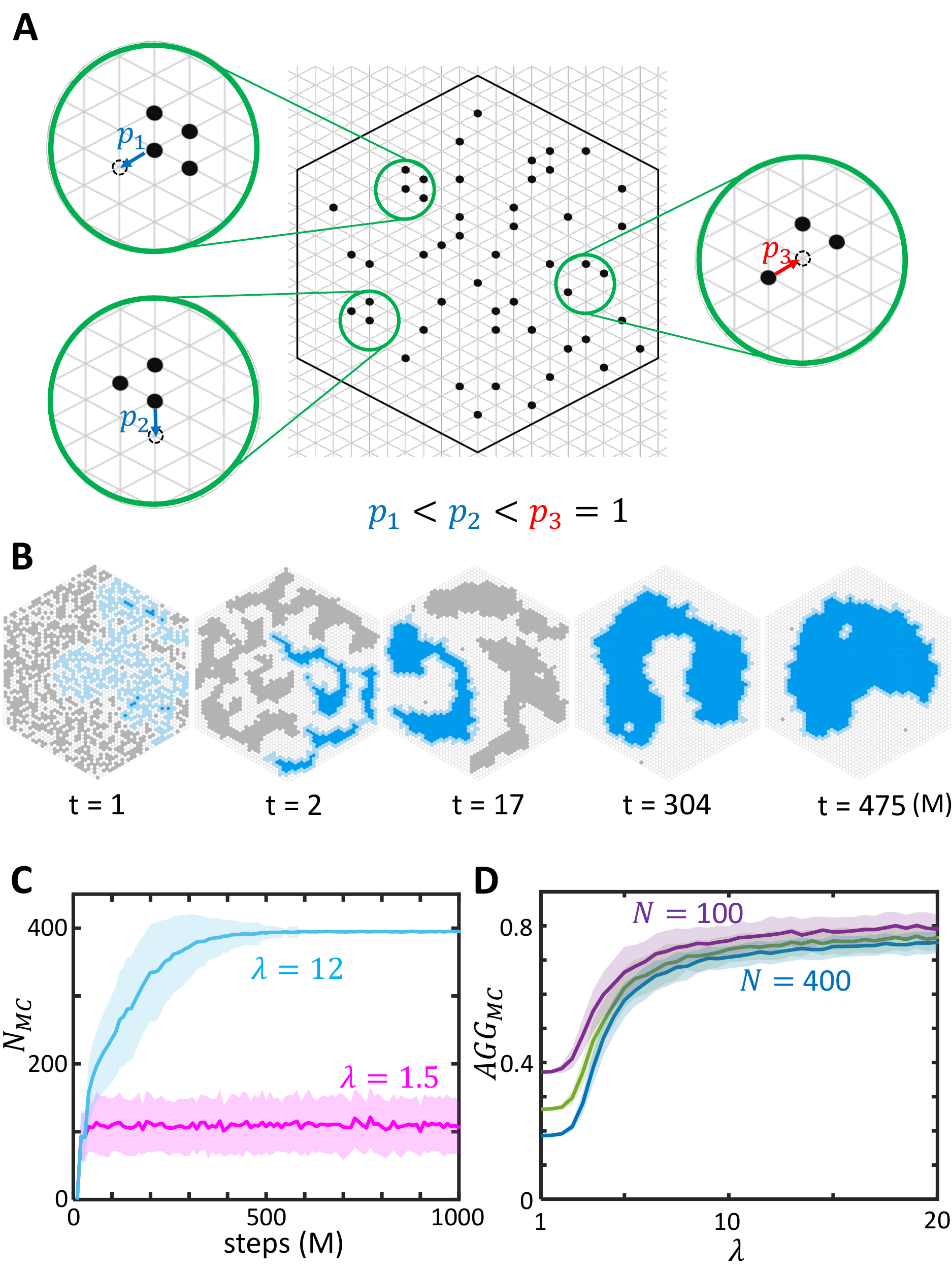}
    \caption{\textbf{The theoretical self-organizing particle system (SOPS).}
    (\textbf{A}) A particle moves away from a node where it has $n$ neighbors with probability ${\lambda^{-n}}$, where $\lambda > 0$.
    Thus, moves from locations with more neighbors are made with smaller probability than those with fewer (e.g., in the insets, $p_1 = \lambda^{-3} < p_2 = \lambda^{-2} < p_3 = 1$).
    (\textbf{B}) Time evolution of a simulated SOPS with 1377 particles for $\lambda = 7.5$ showing progressive aggregation (Movie S1).
    The bulk of the largest connected component is shown in blue and its periphery is shown in light blue.
    (\textbf{C}) Time evolution of $N_{MC}$, the size of the largest connected component, showing dispersion for $\lambda = 1.5$ and aggregation for $\lambda = 12$.
    The simulations use $400$ particles.
    (\textbf{D}) Phase change in $\lambda$-space for the aggregation metric $AGG_{MC} = N_{MC} / (k_0P_{MC}\sqrt{N})$, where $k_0$ is a scaling constant, $P_{MC}$ is the number of particles on the periphery of the largest component, and $N$ is the total number of particles.
    This phase change is qualitatively invariant to the system's size.}
    \label{fig:SOPS}
\end{SCfigure}

Let $\Omega$ be the set of configurations with $N$ particles within our bounded lattice region.
We will use the following definition to quantify aggregation for particles that can be disconnected, capturing both the size and compactness of aggregates.
\begin{definition}
	For $\beta > 0$ and $\delta \in (0, 1/2)$, a configuration $\sigma \in \Omega$ is \underline{$(\beta,\delta)$-aggregated} if there is a subset $R$ of lattice nodes such that:
	\begin{enumerate}
		\item At most $\beta \sqrt{N}$ edges have exactly one endpoint in $R$;
		\item The density of particles in $R$ is at least $1-\delta$; and
		\item The density of particles not in $R$ is at most $\delta$.
	\end{enumerate}
\end{definition}
\noindent Here, $\beta$ is a measure of how  small the boundary between $R$ and its complement $\overline{R}$ must be, measuring the compactness of the aggregated particles, and $\delta$ is a tolerance for having unoccupied nodes within the cluster $R$ or occupied nodes outside of $R$.
We say that a configuration is {\it dispersed} if no such $(\beta, \delta)$ exist.

By carefully analyzing the stationary distribution of $\cal{M}_{\rm AGG}$, which is just the desired Boltzmann distribution, we establish conditions that provably yield aggregation when the particles are confined to a compact region of the triangular lattice (\figtext~\ref{fig:SOPS}B).
The proof uses arguments from \cite{Cannon2019}; see the Materials and Methods for details.

\begin{theorem} \label{thm:aggregation}
    Let configuration $\sigma$ be drawn from the stationary distribution of $\cal{M}_{\rm AGG}$ on a bounded, compact region of the triangular lattice, when the number of particles $N$ is sufficiently large.
    If $\lambda > 5.66$, then with high probability there exist $\beta > 0$ and $0 < \delta < 1/2$ such that $\sigma$ will be $(\beta, \delta)$-aggregated.
    However, when $0.98 < \lambda < 1.02$, the configuration $\sigma$ will be dispersed with high probability.
\end{theorem}

Varying values of $\lambda$ in simulation gives strong indication that dispersion persists for larger values of $\lambda$ and the aggregation algorithm undergoes a phase transition whereby the macroscopic behavior of the system suddenly changes from dispersion to aggregation (\figtext~\ref{fig:SOPS}C--D, Movie S1), mimicking the fixed magnetization ferromagnetic Ising model which motivated our Markov chain algorithm.
Nonetheless, our proofs demonstrate that our system has two distinct phases of behavior for different ranges of $\lambda$ for any system with a sufficiently large number of interacting particles, which is enough for our purposes.

\subsection*{BOBbots: a model active cohesive granular matter system}

Next, to test whether the lattice-based equilibrium system can be used to control a real-world swarm in which there are no guarantees of detailed balance or Boltzmann distributions, we introduce a collective of active cohesive granular robots which we name \textit{BOBbots} (\figtext~\ref{fig:bobbots}A--C,~\ref{fig:bobbotdesign}) --- Behaving, Organizing, Buzzing robots --- whose design \textit{physically embodies} the aggregation algorithm. 
Driven granular media provide a useful soft matter system to integrate features of the physical world into the toolkit for programming collectives.
This builds upon three decades of work understanding how forced collections of simple particles interacting locally can lead to remarkably complex and diverse phenomena, not only mimicking solids, fluids, and gasses~\cite{andreotti2013granular,lim2019cluster} --- e.g., in pattern formation~\cite{melo1995hexagons,eshuis2007phase}, supercooled and glassy phenomena~\cite{keys2007measurement,goldman2006signatures}, and shock waves~\cite{rericha2001shocks} --- but also displaying phenomena characteristic of soft matter systems such as stress chains~\cite{howell1999stress} and jamming transitions~\cite{corwin2005structural,bi2011jamming}.
While cohesive granular materials are typically generated in situations where particles are small (powders, with interactions dominated by electrostatic or even van der Waals interactions) or wet (with interactions dominated by formation of liquid bridges between particles)~\cite{Mitarai2006-wetgranular,Hemmerle2016-cohesivegranular}, we generate our cohesive granular robots using loose magnets which can rotate to always achieve attraction.

\begin{SCfigure}[][t]
    \centering
    \includegraphics[width=0.65\linewidth]{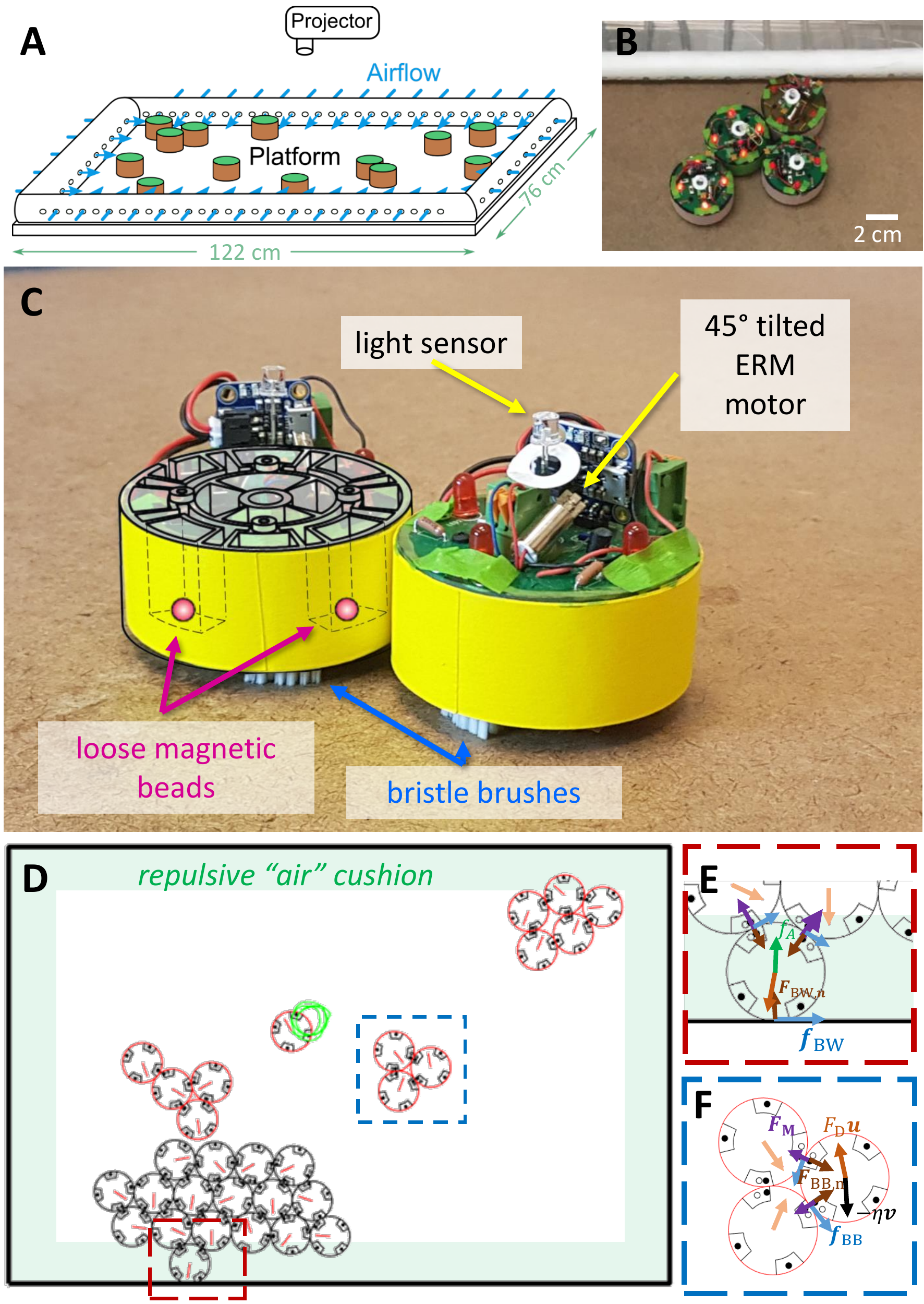}
    \caption{\textbf{BOBbots and their collective motion.}
    (\textbf{A}) Schematic of experimental setup. BOBbots are placed in a level arena with airflow gently repelling them from the boundaries.
    (\textbf{B}) A closeup of the experimental platform.
    (\textbf{C}) Mechanics of the BOBbots.
    Loose magnetic beads housed in the BOBbots' peripheries can reorient so BOBbots always attract each other.
    The vibration of the ERM motor and the asymmetry of bristles lead to the directed motion.
    The light sensor activates the motion.
    (\textbf{D}) Discrete element method simulation setup.
    (\textbf{E}) The BOBbot-boundary interactions: airflow repulsion $f_A$, BOBbot-boundary friction $f_{\text{BW}}$, and normal force $F_{\text{BW,n}}$.
    (\textbf{F}) The inter-BOBbot interactions: attraction between magnetic beads $F_M$, inter-BOBbot friction $f_{\text{BB}}$, and sterical exclusion $F_{\text{BB,n}}$.}
    \label{fig:bobbots}
\end{SCfigure}

The movement and interactions between BOBbots were designed to capture the salient features of the abstract stochastic algorithm while replacing all sensing, communication, and probabilistic computation with physical morphology and interactions.
Each BOBbot has a cylindrical chassis with a base of elastic ``brushes'' that are physically coupled to an off-center eccentric rotating mass vibration motor (ERM).
The vibrations caused by the rotation of the ERM are converted into locomotion by the brushes (\figtext~\ref{fig:bobbots}C).
Due to asymmetry in our construction of this propulsion mechanism, the BOBbots traverse predominantly circular trajectories \cite{kummel2013circular} that are randomized through their initial conditions but --- unlike the SOPS particles --- are inherently deterministic with some noise and occur at a constant speed per robot distributed as $v_0 = 4.8 \pm 2.0$ cm/s.
See the Materials and Methods for further details.

Analogous to the modified transition probabilities in the aggregation algorithm that discourage particles from moving away from positions where they have many neighbors, each BOBbot has loose magnets housed in shells around its periphery that always reorient to be attractive to nearby BOBbots (\figtext~\ref{fig:bobbots}C).
The probability that a BOBbot detaches from its neighbors is negatively correlated with the attractive force from the number of engaged magnets, approximating the movement probabilities given by the algorithm which scale inversely and geometrically with the number of neighbors.
We subsequently verify this assertion experimentally (see Section S5 of the Supplementary Materials for details).
The strength of the magnets $F_{M0}$ determines whether the system aggregates or disperses in the long run, analogous to $\lambda$ in the algorithm.

To allow for study of larger BOBbot ensembles and more comprehensive sweeps of parameter space, we also performed Discrete-Element Method (DEM) simulations of the BOBbots (see \figtext~\ref{fig:bobbots}D--F and the Materials and Methods for more details).
The motion of an individual BOBbot is modeled as a set of overdamped Langevin-type equations governing both its translation and rotation subject to its diffusion, drift~\cite{Jahanshahi2017}, magnetic attraction, and sterical exclusion with other BOBbots.
The translational drift corresponds to the speed from the equilibrium of the drive and drag forces while the rotational drift corresponds to the circular rotation.
Similar methods have been used to understand macroscale phenomena emerging from collectives of microscopic elements~\cite{Hamann2018-swarmrobotics} and to model particle motion in active matter~\cite{Ramaswamy2017}.

Mitigating the effects of the arena's fixed boundaries in both experiments and simulations presented a significant design challenge.
BOBbots can persist along the boundary or in corners, affecting system dynamics by, for example, enabling aggregates to form where they would not have otherwise or hindering multiple aggregates from integrating.
To address these issues, uniform airflow was employed to gently repel BOBbots away from the boundary and similar effects were implemented in simulation.
More details about the experimental apparatus and protocol can be found in the Materials and Methods.

\subsection*{Clustering dynamics explained by algorithm analysis}

Since the critical elements of the SOPS algorithm can be physically embodied by robots as simple as our BOBbots, to test if the SOPS model could quantitatively capture collective dynamics, we next investigated the degree to which collectives of BOBbots aggregate as a function of their peripheral magnet strength $F_{M0}$ in both robotic experiments and DEM simulations.
(For convenience, $F_{M0}$ is normalized by the gravity of Earth $g = 9.81$ m/s$^2$ when using the unit of gram.)
The experimental protocol begins with placing magnets of a particular strength $F_{M0}$ into the BOBbots' peripheral slots.
The BOBbots are positioned and oriented randomly in a rectangular arena and are then actuated uniformly for a fixed time during which the BOBbots' positions and the size of the largest connected component are tracked (\figtext~\ref{fig:clustering}A--C).
These trials are conducted for several $F_{M0}$ values with repetition.
We followed the same protocol in simulations.

\begin{SCfigure}[][th]
    \centering
    \includegraphics[width=0.57\linewidth]{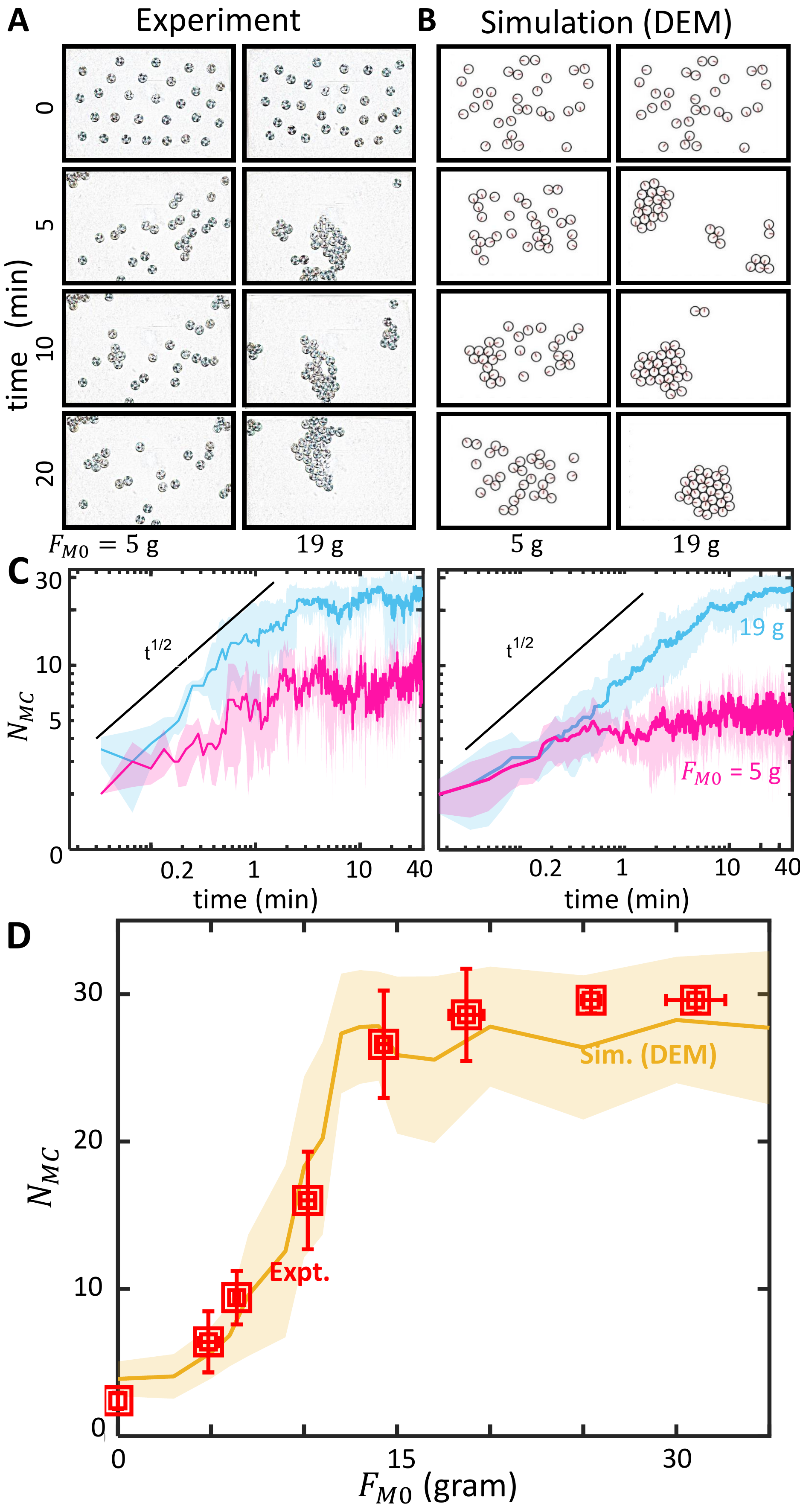}
    \caption{\textbf{Evolution of BOBbot clusters.}
    (\textbf{A}) Time evolution snapshots of both experiment (Movie S3) and (\textbf{B}) simulation (Movie S4) for a system of 30 BOBbots with different magnet strengths: $F_{M0} = 5$ g (left) where we observe dispersion, and $F_{M0} = 19$ g (right) where we observe aggregation.
    Experimental images have been processed with a low-pass filter for better visual clarity.
    (\textbf{C}) Time evolutions of the size of the largest component $N_{MC}$ in experiment and simulation for a system of 30 BOBbots with $F_{M0} = 5$ g (magenta) and $F_{M0} = 19$ g (blue).
    (\textbf{D}) Scaling of cluster size vs.\ magnetic strength for a system of 30 BOBbots showing an increase in $N_{MC}$ as the magnet strength $F_{M0}$ increases.
    The yellow plot line shows the mean and standard deviation of $N_{MC}$ in the $150$ simulation runs for each magnetic strength $F_{M0}$ between $1$--$35$ g, with a step size of $1$ g.
    Experimental data is shown in red with error bars showing the standard deviation of largest cluster size $N_{MC}$ and the uncertainty of $F_{M0}$ due to empirical measurement.}
    \label{fig:clustering}
\end{SCfigure}

In experiment and DEM simulation, we observe an abrupt, rapid rise and then saturation in the size $N_{MC}$ of the largest connected component as the magnetic attraction $F_{M0}$ increases (\figtext~\ref{fig:clustering}D).
These curves resemble those in \figtext~\ref{fig:SOPS}D, with the magnetization $F_{M0}$ playing a role analogous to the bias parameter $\lambda$.
Given this correspondence, we explored whether the equilibrium SOPS model could be used to make quantifiable predictions in the robot experiments.
First, we designed a test to examine how force and $\lambda$ scale.
Recall that in the SOPS algorithm, the force acting on each particle is proportional to $\lambda^n$, where $n$ is the particle's current number of neighbors.
In the experiments, BOBbots cannot count their neighbors, but the magnets are expected to provide a similar force that also increases geometrically when more magnets are engaged.

To estimate the relationship between force and $\lambda$, we investigate the rate at which a BOBbot loses or gains neighbors over a fixed amount of time. 
Viewing a BOBbot's completion of half its circular motion as analogous to a particle moving to a new lattice node in the SOPS algorithm and using this time interval to evaluate the transition, simulation data shows that a BOBbot's transition probability from having a higher number of neighbors $n$ to a lower number $n'$ closely follows the algorithm's $P(\sigma, \tau) \propto \min(1, \lambda^{n'-n})$ transition probabilities (\figtext~\ref{fig:alganalysis}A,~\ref{fig:detachprob}).
Further, we evaluated the BOBbots' effective bias parameter $\lambda_{\text{eff}}$ as a function of $F_{M0}$ and found an exponential relation $\lambda_{\text{eff}} = \exp(\beta F_{M0})$, where $\beta$ is a constant representing inverse temperature (\figtext~\ref{fig:alganalysis}B).
The BOBbots' transition probabilities can then be approximated as $P(\sigma, \tau) = \exp(-\beta(\epsilon_n - \epsilon_{n'}))$, where $\beta$ is the inverse temperature of the system and $\epsilon_n = n \cdot F_{M0}$ can be interpreted as the energy contributed by a BOBbot's $n$ neighbors.

\begin{SCfigure}[][th]
    \centering
    \includegraphics[width=0.58\linewidth]{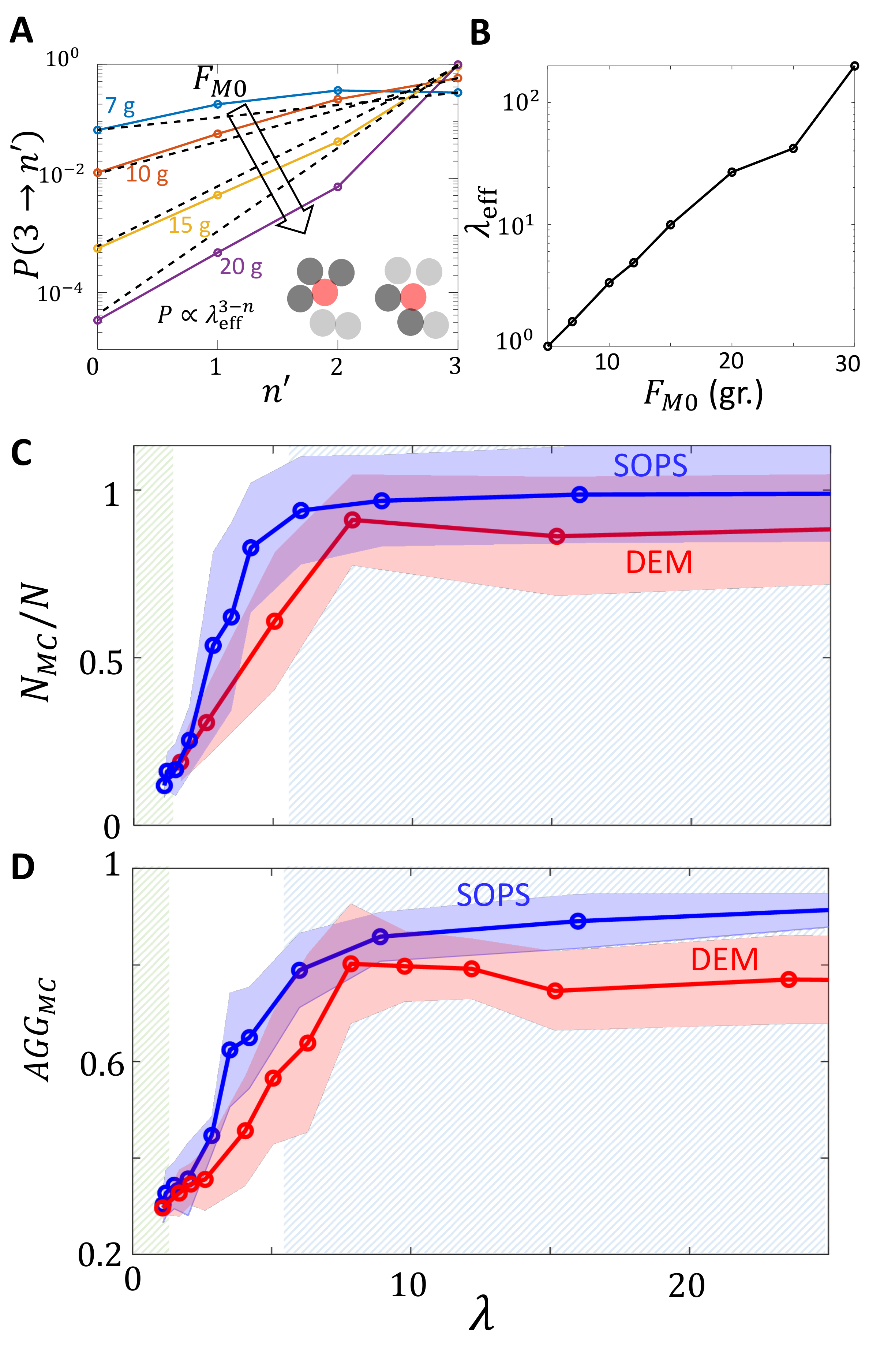}
    \caption{\textbf{Algorithmic interpretation of BOBbot clustering.}
    (\textbf{A}) A diagram showing how the effective bias parameter $\lambda_\text{eff}$ is evaluated from the DEM simulation.
    (\textbf{B}) The dependence of $\lambda_\text{eff}$ on the magnetic attraction force $F_{M0}$.
    The (\textbf{C}) maximum cluster fraction $N_{MC}/N$ and (\textbf{D}) aggregation metric $AGG_{MC}$ for different values of $\lambda$ in both the SOPS algorithm (blue) and physical simulations (red).
    The green and blue shaded regions show the dispersed and aggregated regimes proved from theory, respectively.}
    \label{fig:alganalysis}
\end{SCfigure}

With the relation between $F_{M0}$ and $\lambda_{\text{eff}}$ established, we next compare the aggregation behaviors exhibited by the SOPS algorithm and the BOBbot ensembles.
\figtext~\ref{fig:alganalysis}C shows the fraction of particles/BOBbots in the largest component $N_{MC}/N$ observed in both the SOPS algorithm and BOBbot simulations after converting with respect to $\lambda_{\text{eff}}$; the algorithm does indeed capture the maximum cluster fraction observed in the simulations.
Notably, the aggregated and dispersed regimes in $\lambda$-space established by Theorem~\ref{thm:aggregation} provide a rigorous understanding of these BOBbot collective behaviors.
For instance, the proven dispersed regime $0.98 < \lambda < 1.02$ gives a clear explanation for why agents will not aggregate even in the presence of mutual attraction.
Further, it also helps establish the magnitude of attraction needed to saturate the aggregation.

\begin{SCfigure}[][ht]
    \centering
    \includegraphics[width=0.5\linewidth]{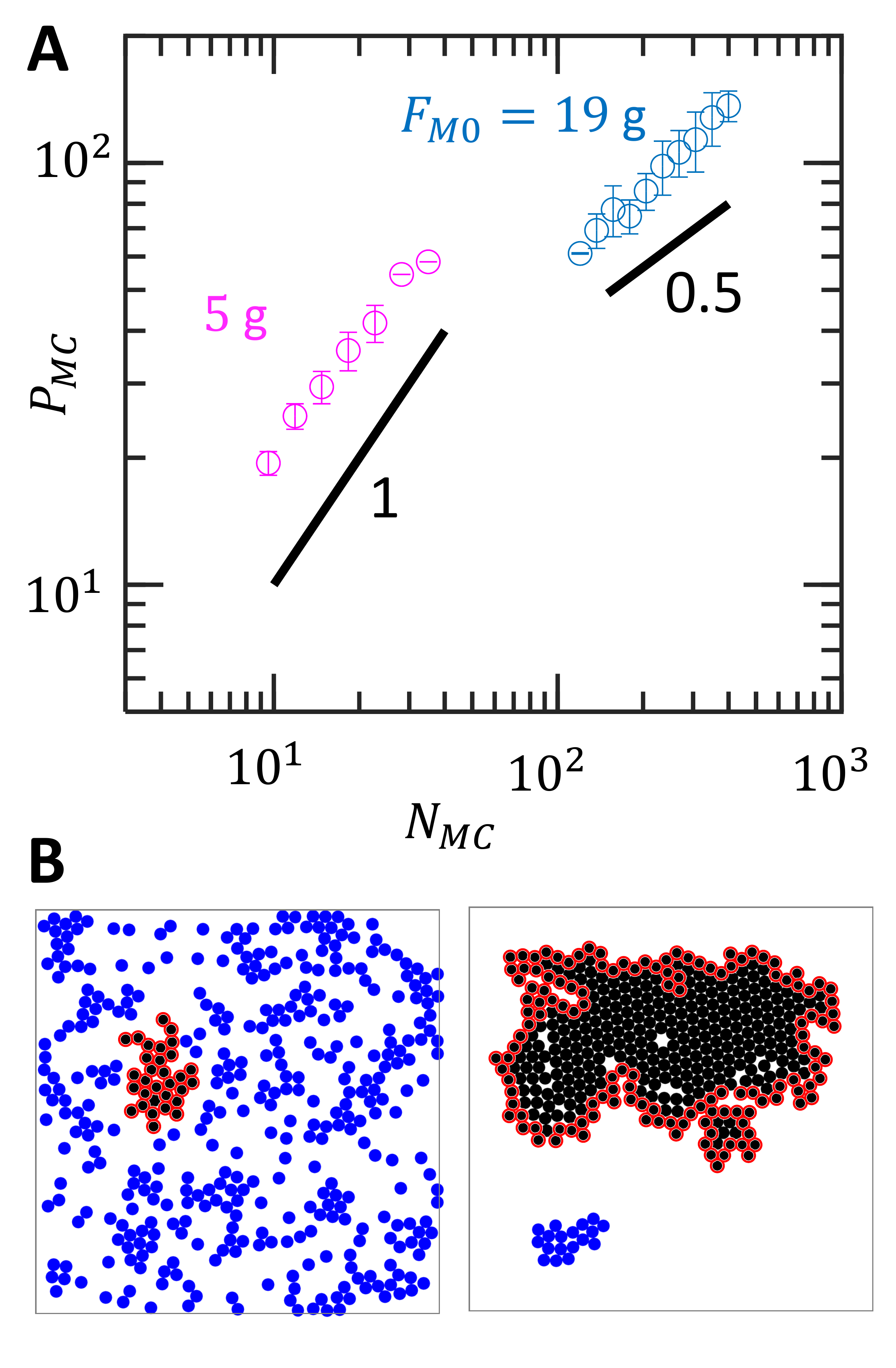}
    \caption{\textbf{Perimeter scaling of BOBbot clusters.}
    (\textbf{A}) Log-log plot showing the scaling relationship between the largest component's size $N_{MC}$ and perimeter $P_{MC}$ in number of BOBbots for simulated systems of 400 BOBbots with $F_{M0} = 5$ g (magenta) and $19$ g (cyan) for fixed boundary conditions.
    Each data point is the average of 20 simulations.
    While the SOPS predicts a scaling power of $0.5$ for the aggregated case (cyan), the data shows a slightly larger --- but still sublinear --- power of $0.66 \pm 0.07$.
    (\textbf{B}) Final snapshot of the collective motion of 400 BOBbots with $F_{M0} = 5$ g (left) and $19$ g (right).
    BOBbots shown in black belong to the largest connected component; those outlined in red are on its periphery.}
    \label{fig:alganalysisB}
\end{SCfigure}

We additionally test the SOPS prediction that the maximum cluster should not only be large but also compact, occupying a densely packed region.
The results from~\cite{Cannon2019} that we apply here for aggregation suggest the following relationship between the size of the largest component $N_{MC}$ and its perimeter $P_{MC}$.
In dispersed configurations, $P_{MC}$ should scale linearly with $N_{MC}$, meaning that most BOBbots lie on the periphery of their components.
In aggregated configurations, however, $P_{MC}$ should scale as $N_{MC}^{1/2}$, approximating the minimal perimeter for the same number of BOBbots by at most a constant factor.
We test these scaling relationships in simulations with 400 BOBbots (\figtext~\ref{fig:alganalysisB}A) and find that the theory's predictions hold in the dispersed regime; however, the $0.66 \pm 0.07$ sublinear scaling power for the aggregated case is slightly higher than the theory's prediction of $0.5$.
This discrepancy may in part be due to boundary and finite-size effects --- in fact, DEM simulations with periodic boundaries show a scaling power of $0.59 \pm 0.18$ that is closer to the SOPS theory (\figtext~\ref{fig:thermoDynLim}) --- but is also affected by non-reversibility inherent in the BOBbots' circular trajectories.
To make a quantitative comparison that captures when components are both large and compact, we track $AGG_{MC} = N_{MC} / (k_0 P_{MC} \sqrt{N})$, where $k_0$ is a scaling constant defined such that $AGG_{MC} = 1$ when the system is optimally aggregated, achieving the minimum possible perimeter.
Physically, $AGG_{MC}$ is reminiscent of surface tension for which energy minimization leads to a smaller interface (in our setting, smaller perimeter $P_{MC}$), yielding an $AGG_{MC}$ closer to 1.
We obtain agreement between the SOPS and DEM simulations with respect to this metric as well (\figtext~\ref{fig:alganalysis}D), further validating the theory's prediction, though the DEM simulations yield slightly smaller $AGG_{MC}$ than the SOPS algorithm for large $\lambda$.

We noticed the size of the largest component $N_{MC}$ grows roughly proportional to $t^{1/2}$ over time (\figtext~\ref{fig:clustering}C).
Since the perimeter of the largest cluster $P_{MC}$ scales proportional to $N_{MC}^{0.66}\approx N_{MC}^{2/3}$ (\figtext~\ref{fig:alganalysisB}A), this implies the length scale grows like $t^{1/3}$.
This is reminiscent of coarsening in a broad class of systems described by Cahn--Hilliard equation $\partial u/ \partial t = \nabla^2 (\Phi'(u)-\gamma \nabla^2 u)$ where order parameter $u$ takes continuous values in $(-1, 1)$ where $-1$ and $1$ are analogous to empty and occupied nodes in the SOPS lattice, respectively.
To bridge the SOPS algorithm with the Cahn--Hilliard equation, we first observe that the SOPS algorithm with bias parameter $\lambda$ can be exactly mapped to an Ising model with fixed magnetization~\cite{jerrum1993polynomial,randall1999sampling} with coupling strength $J = \frac{1}{2\beta}\log\lambda$, where $\beta$ is inverse temperature (see Section S7 of the Supplementary Materials for details).
As shown by Penrose~\cite{penrose1991mean}, the fixed magnetization Ising model with coupling strength $J$ can be mapped to the surface tension $\gamma$ of the Cahn--Hilliard equation as $\gamma = \beta J$.
Thus, the SOPS and BOBbot ensemble behaviors map to the Cahn--Hilliard equation with $\gamma = \frac{1}{2}\log\lambda \propto F_{M0}$.
This suggests that in the limit, the SOPS and BOBbot aggregation behavior should display a second-order phase transition at a critical $\lambda_c$ corresponding to the critical surface tension $\gamma$ in the Cahn--Hilliard equation.
The corresponding critical value $\lambda_c = e^{2/7} \approx 1.33$ on the hexagonal lattice lies within the $\lambda_c \in (1.02, 5.66)$ range proven by the SOPS theory (see Section S8 of the Supplementary Materials for details).
Thus, we obtain agreement between the SOPS theory for a finite lattice system and the Cahn--Hilliard equation for an active matter system at the continuum limit.
This mapping gives further confirmation of the universality of our results and provides another perspective for ``programming'' active collectives.

\begin{figure}[th]
    \centering
    \includegraphics[width=0.6\linewidth]{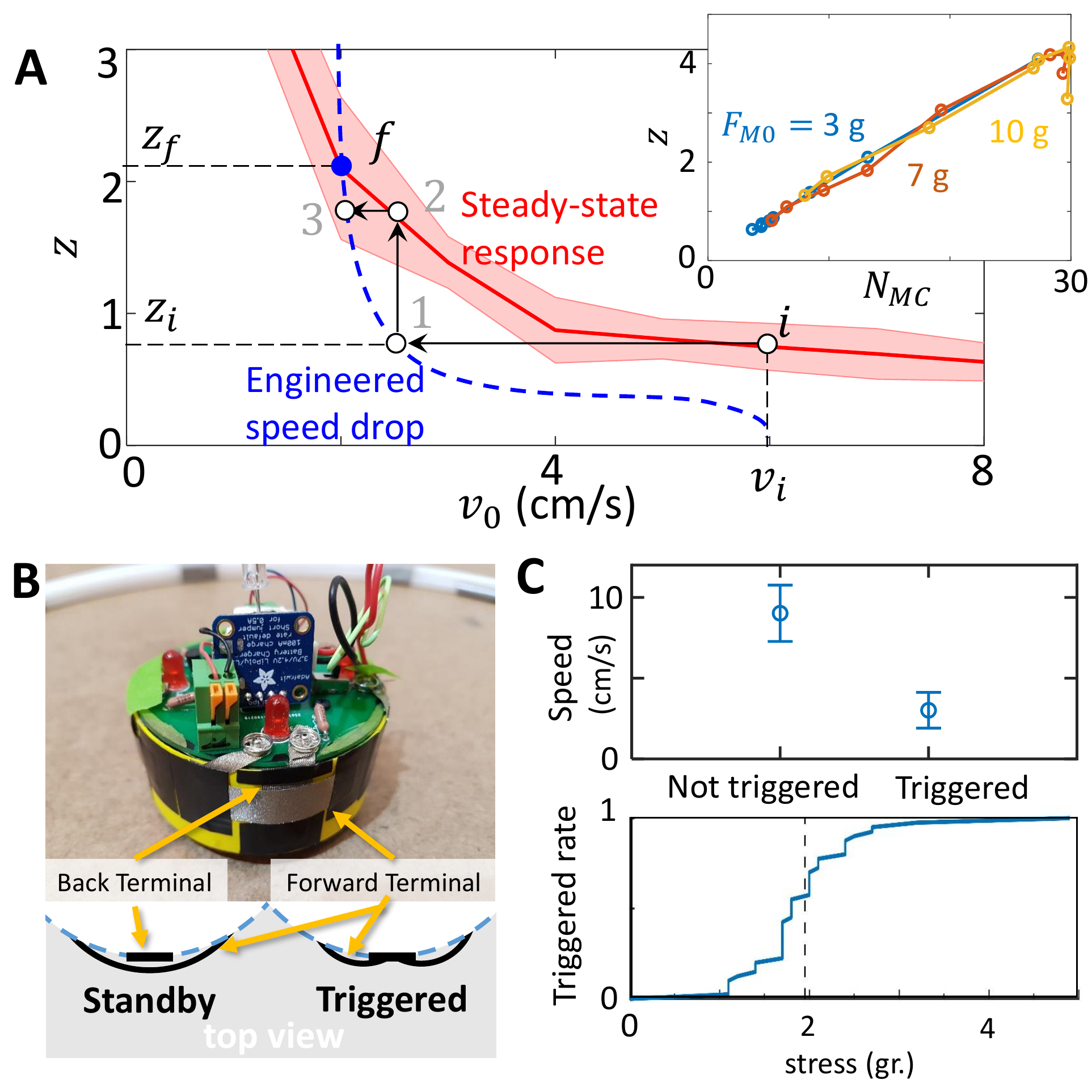}
    \caption{\textbf{Design and implementation of stress sensing for enhanced aggregation.}
    (\textbf{A}) The effect of the engineered, adaptive speeds (blue) on the steady-state average number of neighbors per BOBbot (red) for $F_{M0} = 3$ g.
    Without adapting speeds, BOBbots actuated at a given speed $v_i$ would obtain an average of $z_i$ neighbors per BOBbot at equilibrium (initial point $i$).
    With the adaptive speeds, an average of $z_i$ neighbors per BOBbot causes the average speed to slow ($i \to 1$) which in turn enables convergence to the steady-state response with more neighbors per BOBbot ($1 \to 2$).
    This feedback iterates until the steady-state and engineered responses coincide at final point $f = (v_f, z_f)$, where $v_f < v_i$ and $z_f > z_i$.
    Inset: The mapping between maximum cluster size $N_{MC}$ and the average number of neighbors per BOBbot $z$ indicates the stress-sensing control strategy will increase component sizes.
    (\textbf{B}) A BOBbot equipped with a stress sensor and the schematic top-view sketch of the triggered and not-triggered states.
    (\textbf{C}) The BOBbot's response to stress.
    Top: the speed of a BOBbot when its sensor is and is not triggered.
    Bottom: The rate of sensor triggering as function of the stress applied.}
    \label{fig:stress}
\end{figure}

\subsection*{Enhancing clustering via local stress sensing}

We have demonstrated that the BOBbot ensembles mimic a lattice model that can provably aggregate for large enough $\lambda$, corresponding physically to highly attractive interaction that favors large components with small perimeter.
We now ask whether we can achieve rudimentary collective intelligence determining, for example, how robots could tune their responses to enhance or dampen aggregation, thereby achieving a more tightly clustered or dispersed state.
In particular, we explore whether such tuning can help counteract some ways the system deviates from the theory, such as variations in the BOBbots' speeds and magnetic attraction, improving the fidelity to the original algorithm.
While the BOBbots remain unable to count neighbors or estimate the Gibbs probabilities directly as prescribed by the algorithm, we take advantage of physical effects of the BOBbot ensembles to ``program'' desirable behavior without using any traditional computation.

The first effect relies on observations that for a fixed magnet strength, the size of the largest component $N_{MC}$ decreases with increasing BOBbot speed $v_0$ (\figtext~\ref{fig:NMCvsV}); a full investigation of the behavior of BOBbot collectives at varying uniform speeds will be the subject of a separate study.
We further observe that $N_{MC}$ scales linearly with $z$, the average number of neighbors per BOBbot at equilibrium (\figtext~\ref{fig:stress}A, inset).
Thus, BOBbot speed $v_0$ is inversely correlated with the average number of neighbors per BOBbot $z$.
This arises from $v_0$ being a proxy for $\beta^{-1}$ in the effective attraction $\lambda_{\text{eff}}$.
Consequently, we can mimic enhanced aggregation via increased magnet strength by reducing a BOBbot's speed as a function of its number of neighbors.

Without adapting a BOBbot's speed based on its number of neighbors, a BOBbot collective actuated uniformly at a speed $v$ converges to an average of $z_\text{std}(v)$ neighbors per BOBbot at equilibrium (\figtext~\ref{fig:stress}A, red); any point in speed-neighbor space deviating from $z_\text{std}(v)$ is transient.
To enhance aggregation, we engineer reduced speeds $v_\text{eng}(z)$ that a BOBbot with $z$ neighbors should adapt to (\figtext~\ref{fig:stress}A, blue).
These slowed speeds allow the collective to reconverge to a new steady-state with a larger number of average neighbors per BOBbot (\figtext~\ref{fig:stress}A, arrows).
This feedback between the engineered speeds $v_\text{eng}$ and the steady-state average number of neighbors $z_\text{std}$ iterates until reaching the fixed point in speed-neighbor space where the steady-state and engineered behaviors meet as $z = z_\text{std}(v_\text{eng}(z))$.

\begin{SCfigure}[][ht]
    \centering
    \includegraphics[width=0.64\linewidth]{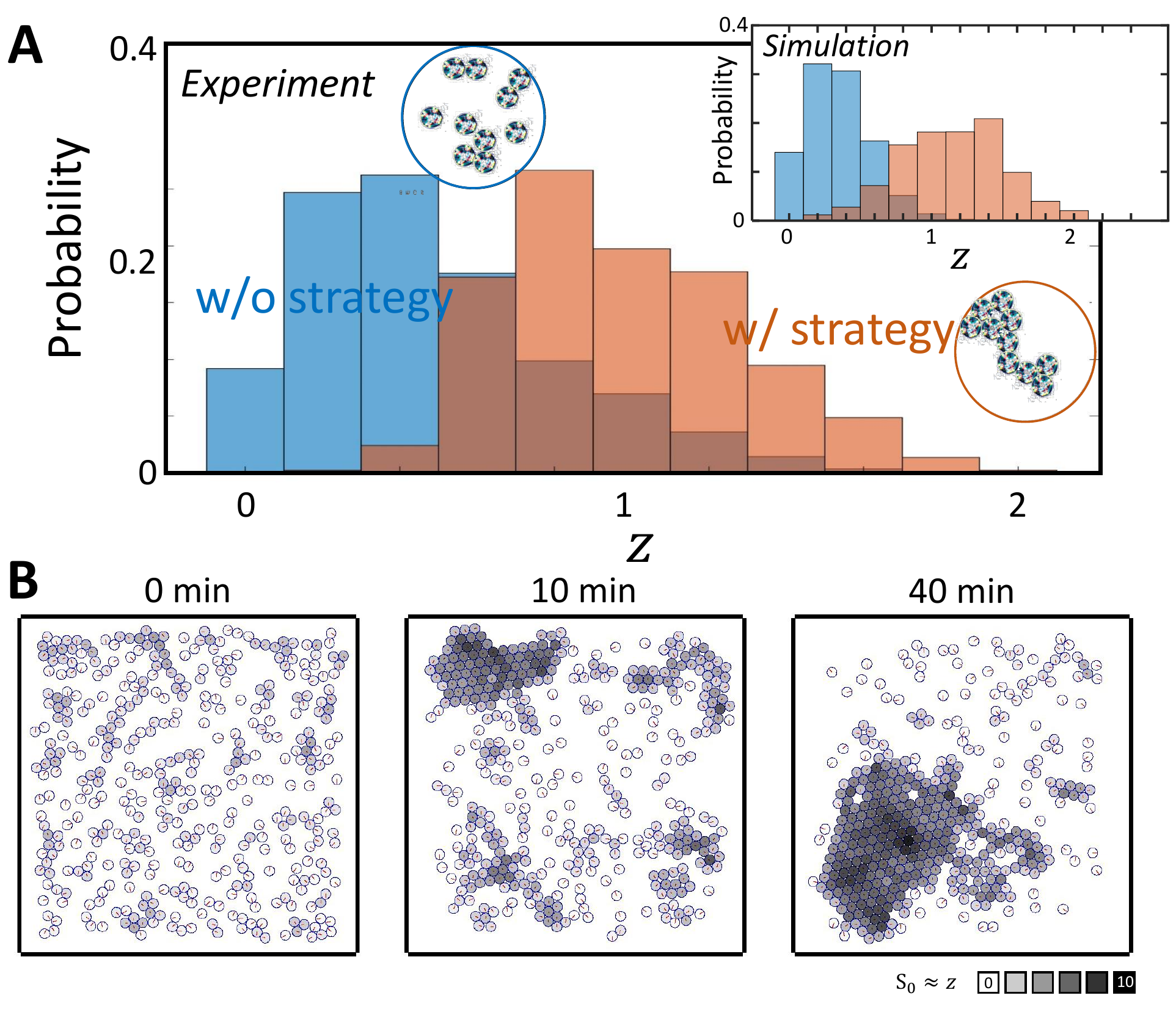}
    \caption{\textbf{Adapting speed via stress sensing enhances aggregation.}
    (\textbf{A}) The distribution of a BOBbot's number of contacts over six 10-minute experiments using $F_{M0} = 3$~g.
    Each sample is an average of number of contacts over 1 second.
    Inset: Simulation results using the same conditions as the experiment.
    (\textbf{B}) A simulation demonstrating enhanced aggregation in an ensemble of 400 BOBbots using a weak magnet strength of $F_{M0} = 7$~g.
    Each BOBbot's speed decreases from $6$ cm/s to $1.2$ cm/s as its stress $s_0 = \sum_{j\in \text{neighbors}} s_j/F_{M0} \approx z$ increases from $0$ to $6$, where $z$ is its current number of neighbors.
    BOBbots in an aggregate's interior experience the most stress (dark gray) and thus have the slowest speeds, enabling larger aggregates to form.
    Without adapting speed in response to stress, the cluster sizes remain the same magnitude as in the 0 minute snapshot (left).}
    \label{fig:stressB}
\end{SCfigure}

While adapting speeds based on numbers of neighbors would be relatively straightforward to implement in more complex robots capable of counting neighbors (e.g., optically as in~\cite{Bayindir2009-aggregation,Rubenstein2014,Valentini2016-collectiveperception,Piranda2018}), implementing such a scheme in the deliberately simple BOBbots is challenging given their lack of such sensing.
Here we utilize a second physical effect: inspired by the correlation of particle density and stress on individual particles in granular systems~\cite{majmudar2007jamming}, we propose that monitoring local contact stress can function as a proxy for counting numbers of neighbors.
An immediate benefit of such a scheme is that it can be implemented on the existing robots via custom, low-cost, analog surface stress sensors (see \figtext~\ref{fig:stress}B and the Materials and Methods for details).
The implemented stress sensors function such that for sufficiently large stress (e.g., when in a cluster), motor speed is decreased by $70\%$ (\figtext~\ref{fig:stress}C).

We implemented this ``physical algorithm'' on BOBbot ensembles with weakly attractive magnets (Movie S6).
In experiments with ensembles of 10 BOBbots in a circular arena, adapting BOBbot speeds in response to stress sensing significantly increases the average number of neighbors per BOBbot (\figtext~\ref{fig:stressB}A).
Further, there is a quantitative match in the final average number of neighbors per BOBbot between the experiments and the fixed points predicted in \figtext~\ref{fig:stress}A, validating our control strategy for enhancing aggregation.
Simulations using the same arena and stress-mediated response reproduce the experimental results (\figtext~\ref{fig:stressB}A, inset).
In simulations of 400 BOBbots with $F_{M0} = 7$~g, we observe that BOBbots with more neighbors experience higher stress and thus have the slower speeds (\figtext~\ref{fig:stressB}B).
This stress-mediated decrease in speed enables large aggregates to form that would not have existed otherwise in the weakly attractive regime.
The use of stress sensing opens an interesting avenue for collectives of rudimentary robots to incorporate higher-order information without complex vision systems; further, contact stress provides insights (e.g., closeness to a jamming transition) that could be valuable in densely packed clusters~\cite{aguilar2018collective}.

\begin{SCfigure}[][th]
    \centering
    \includegraphics[width=0.5\linewidth]{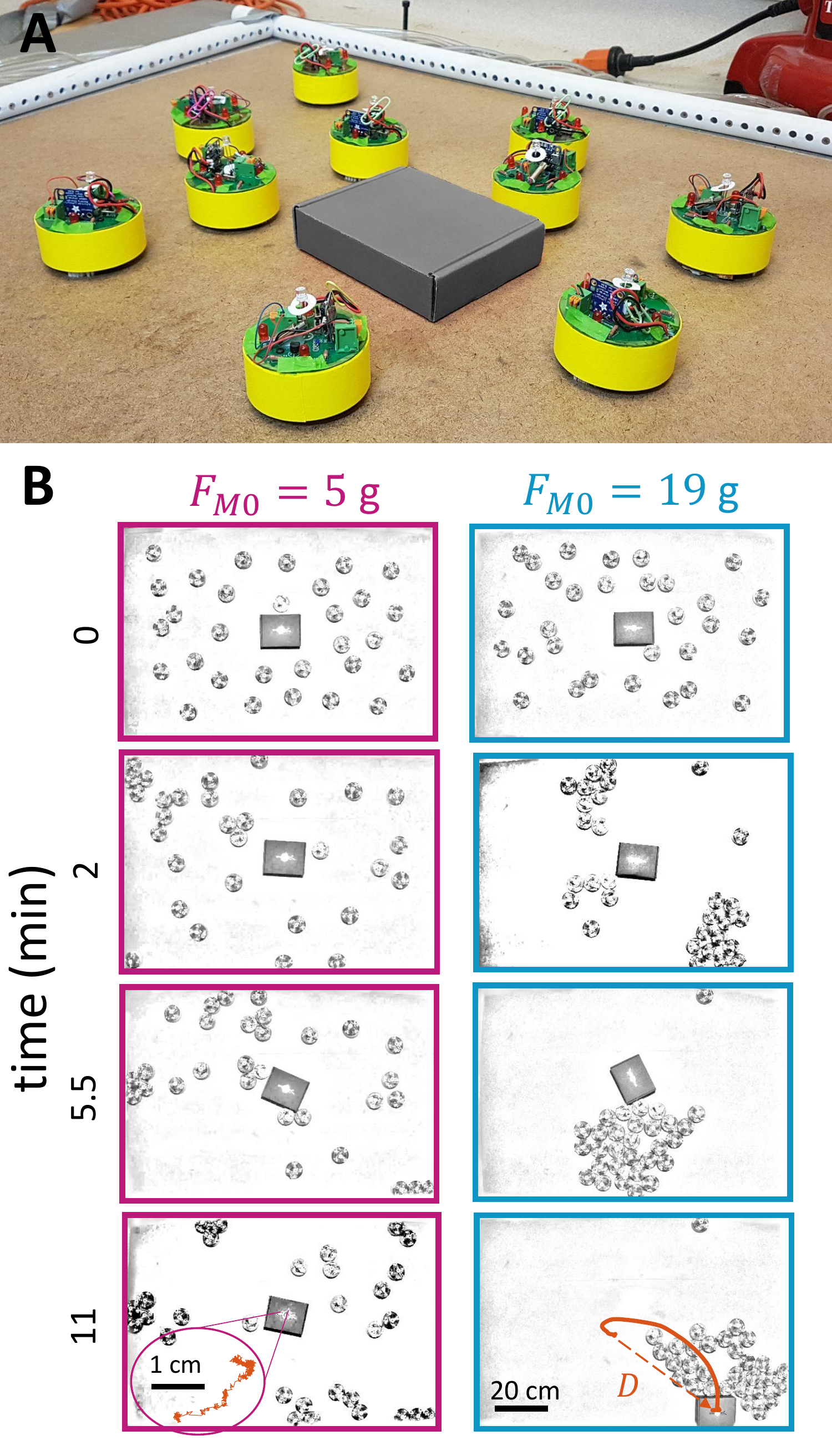}
    \caption{\textbf{Object transport using aggregation.}
    (\textbf{A}) Schematic of the experimental setup.
    (\textbf{B}) Time evolution snapshots of box transport by a system of 30 BOBbots with magnet strength $F_{M0} = 5$ g and $19$ g (Movie S7).
    The box has a mass of $60$ g.
    The final panel shows the object's complete trajectory, where $D$ denotes the Euclidean distance of the final displacement.}
    \label{fig:transport}
\end{SCfigure}

\subsection*{Object transport in the aggregated phase}

Encouraged by the close connections between the physical system and the underlying theoretical model along with the successful control scheme for enhanced aggregation using stress sensing, we sought to test whether aggregated BOBbots could collectively accomplish a task.
In particular, could an aggregated BOBbot collective ``recognize'' the presence of a non-robot impurity in its environment and cooperatively expel it from the system?
Typically, such collective transport tasks --- e.g., the cooperative transport of food by ants~\cite{Wilson2014-stochastictransport,Feinerman2018} --- either manifest from an order-disorder transition or rely heavily on conformism between agents for concerted effort and alignment of forces.
With our BOBbot collectives, we instead aim to accomplish transport via simple mechanics and physical interactions emergently controlling global behavior without any complex control, communication, or computation.

By maintaining a high magnetic attraction $F_{M0}$, we remain in the aggregated regime where most BOBbots connect physically and can cumulatively push against untethered impurities (e.g., a box or disk) introduced in the system (\figtext~\ref{fig:transport}A, Movie S7).
The BOBbot collective's constant stochastic reconfiguration grants it the ability to envelop, grasp, and dislodge impurities as their individual forces additively overcome the impurities' friction, leading to large displacement in the aggregated regime (\figtext~\ref{fig:transport}B, right) with a median displacement of $7.9$ cm over 12 minutes.
On the contrary, we find that systems with weak magnetic attraction (i.e., those in the dispersed regime) can typically only achieve small impurity displacement (\figtext~\ref{fig:transport}B, left) with a median displacement of $0.9$ cm over $12$ minutes (see \figtext~\ref{fig:transporttrajectories} for distributions).
We observe infrequent anomalies in which dispersed collectives achieve larger displacement than aggregated ones, but these outliers arise from idiosyncrasies of our rudimentary robots (e.g., an aggregated cluster of BOBbots may continuously rotate in place without coming in contact with an impurity due to the BOBbots' individual orientations in the aggregate; see Movie S7).

Characterizing the impurity's transport dynamics as mean-squared displacement over time $\langle r^2(\tau)\rangle = v\tau^\alpha$ reveals further disparities between the aggregated and dispsered BOBbot collectives (\figtext~\ref{fig:transportB}A).
On a log-log plot, the intercept indicates $\log(v)$, where $v$ is the characteristic speed of the impurity's transport; we observe that in all but one fringe case the strongly attractive collectives achieve transport that is orders of magnitude faster than those of the weakly attractive ones (\figtext~\ref{fig:transportB}B).
The slope of each trajectory indicates the exponent $\alpha$ that characterizes transport as subdiffusive ($\alpha < 1$), diffusive ($\alpha = 1$), or superdiffusive ($\alpha > 1$).
While all the strongly attractive collectives immediately achieve nearly ballistic transport (with $\alpha = 1.85 \pm 0.11$ for $\tau < 20$ s) indicating rapid onset of cluster formation and pushing, the weakly attractive collectives initially exhibit mostly subdiffusive transport (with $\alpha = 0.89 \pm 0.56$ for $\tau < 20$ s) caused by intermittent collisions from the dispersed BOBbots (\figtext~\ref{fig:transportB}C).
When the slight heterogeneous distribution of the dispersed BOBbots remains unchanged for a sufficiently long time, the accumulation of displacement in a persistent direction can cause a small drift, leading to ballistic transport at a longer time scale.
These results align with the predictions of a simple model combining subdiffusive motion with small drift (\figtext~\ref{fig:transporttheory}).
Nonetheless, the transport speeds achieved by the dispersed collectives are two orders of magnitude smaller than those of the strongly attractive ones.

\begin{SCfigure}[][th]
    \centering
    \includegraphics[width=0.6\linewidth]{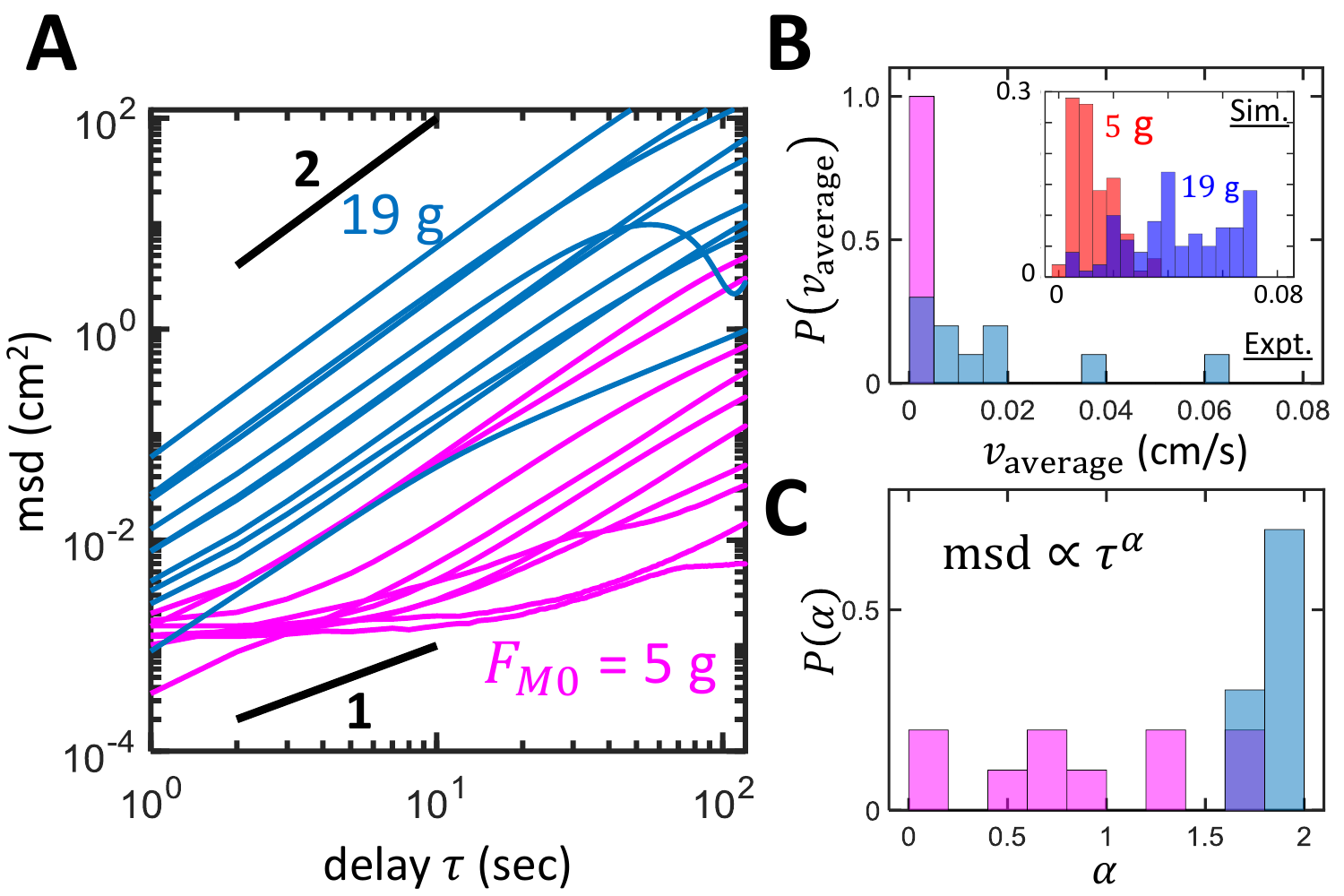}
    \caption{\textbf{Object transport using aggregation.}
    (\textbf{A}) Mean-squared displacement of the box over time in log-log scale for collectives with $F_{M0} = 5$ g (magenta) and $19$ g (blue).
    (\textbf{B}) Distribution of the average speed, calculated as the final displacement $D$ (as shown in Fig.~\ref{fig:transport}B) divided by total time.
    Inset: Simulation results for the overall transport speed.
    The two peaks for $F_{M0} = 19$ g correspond to pushing to the edges and corners.
    (\textbf{C}) Distributions of the mean-squared displacement exponent $\alpha$ at short time scale $\tau < 20$ s.}
    \label{fig:transportB}
\end{SCfigure}

Simulations of impurity transport (see S6 for details) reproduce the experimental results (\figtext~\ref{fig:transportB}B, inset, Movie S7), including the rare anomalies.
Seven of the 100 simulations of weakly attractive collectives succeeded in transporting the impurity to the arena boundary at slow speeds while 76 of the 100 simulations of strongly attractive collectives did so ballistically.
The remaining 24 simulations of attractive collectives that did not achieve ballistic transport consistently formed an aggregate that never came into contact with the impurity.
We found that disaggregating established aggregates by introducing time periods with no attraction enabled them to dissolve and reform for another attempt at transport.
Using different disaggregating sequences, the attractive collectives achieved ballistic transport in 15--20\% more simulations than without disaggregating (\figtext~\ref{fig:shuffling}).
Physically and interestingly, in the Cahn--Hilliard picture, impurity transport can be interpreted as the expulsion of an obstacle in a continuum mixture with sufficiently high surface tension to yield phase separation.
If the obstacle occupies a position that is later occupied by the solid phase, the obstacle is expelled due to sterical exclusion; when its position is unvisited by the solid phase during the process of coarsening, however, it remains stagnant, similar to the anomalies for attractive collectives.
In this interpretation, disaggregating effectively repeats the coarsening process to that the probability any given position is unvisited by the solid phase is significantly diminished.

\section*{Discussion and Conclusion}

In this paper, we use mathematical ideas from distributed computing and statistical physics to create task-oriented cohesive granular media composed of simple interacting robots called BOBbots.
As predicted by the theory, the BOBbots aggregate compactly with stronger magnets (corresponding to large bias parameter $\lambda$) and disperse with weaker magnets (or small $\lambda$).
Simulations capturing the physics governing the BOBbots' motions and interactions further confirm the predicted phase change with larger numbers of BOBbots.
The collective transport task then demonstrates the utility of the aggregation algorithm.

There are several noteworthy aspects of these findings.
First, the \textit{theoretical framework} of the underlying SOPS model can be generalized to allow many types of relaxations to its assumptions, provided its dynamics remain reversible and model a system at thermal equilibrium.
For example, noting that the probability that a robot with $n$ neighbors detaches may not scale precisely as $\lambda^{-n}$ as suggested by the Boltzmann weights, we can generalize the SOPS model to be more sensitive to small variations in these weights: the proofs establishing the two distinct phases can be shown to extend to this setting, provided the probabilities $p_n$ of detaching from $n$ neighbors satisfy $c_1 \lambda^{-n} \leq p_n \leq c_2 \lambda^{-n}$, for constants $c_1, c_2 > 0$.

The \textit{robustness} of the local, stochastic algorithms makes the macro-scale behavior of the collective resistant to many types of idiosyncrasies inherent in the BOBbots, including bias in the directions of their movements, the continuous nature of their trajectories, and nonuniformity in their speeds and magnet strengths.
Moreover, our algorithms are inherently self-stabilizing due to their memoryless, stateless nature, always converging to a desired system configuration --- overcoming faults and other perturbations in the system --- without the need for external intervention.
In our context, the algorithm will naturally continue to aggregate, even as some robots may fail or the environment is perturbed.

We find agreement not only between the BOBbot ensembles and the discrete SOPS model, but also with \textit{continuum models of active matter}.
The SOPS algorithm for aggregation and dispersion was initially defined as a distributed, stochastic implementation of a fixed magnetization Ising model.
In addition to showing that our experimental system follows guarantees established by the analysis of a discrete model, we also observe that the growth of its largest component matches the power-law derived for the Cahn--Hilliard equation, a continuous analog of the Ising model~\cite{penrose1991mean}.
This mapping provides an intuitive understanding of how the SOPS bias parameter $\lambda$, the physical inter-BOBbot attraction $F_{M0}$, and the surface tension $\gamma$ in the Cahn--Hilliard equation correspond; thus, as $\gamma$ controls the phase change in the Cahn--Hilliard equation, so do $\lambda$ and $F_{M0}$ in their respective settings.
This observation buttresses our confidence that the SOPS model provides a useful algorithmic framework capable of producing valid statistical guarantees for ensembles of interacting robots in continuous space.

Moreover, we find that the \textit{nonequilibrium dynamics} of the BOBbots are largely captured by the theoretical models that we analyze at thermal equilibrium, which is in agreement with the findings of Stenhammar et al.~\cite{Stenhammar2013}.
For example, in addition to visually observing the phase change as the magnetic strengths increase, we are able to test precise predictions about the size and perimeter of the largest connected components based on the formal definitions of aggregation and dispersion from the SOPS model.
We additionally use simulations to study the transition probability of a BOBbot from having $n$ neighbors to having $n'$ neighbors to see if the magnetic interactions conform to the theory, and indeed we see a geometric relation decrease in the probability of moving as we increase the number of neighbors, as predicted.
The resultant correspondence between the magnetic attraction and effective bias in the algorithm confirms a quantitative connection between the physical world and the abstract algorithm.

In summary, the framework presented here using provable distributed, stochastic algorithms to inspire the design of robust, simple systems of robots with limited computational capabilities seems quite general.
It also allows one to leverage the extensive amount of work on distributed and stochastic algorithms, and equilibrium models and proofs in guiding the tasks of inherently out of equilibrium robot swarms.
Preliminary results show that we likely can achieve other basic tasks such as alignment, separation (or speciation), and flocking through a similar principled approach.
We note that exploiting physical embodiment with minimal computation seems a critical step in scaling collective behavior to encompass many cutting edge settings, including micro-sized devices that can be used in medical applications and cheap, scalable devices for space and terrestrial exploration.
Additionally, we plan to further study the important interplay between equilibrium and nonequilibrium dynamics to better solidify these connections and to understand which relaxations remain in the same universality classes.

\section*{Materials and Methods}

\subsection*{Details of the SOPS algorithm and proofs}

The SOPS algorithm $\mathcal{M}_{\text{AGG}}$ for aggregation and dispersion is given in Algorithm~\ref{alg:m_agg}.
The algorithm is presented as a Markov chain, but could easily be modified to function as a distributed algorithm executed by each particle independently and concurrently as shown in~\cite{Cannon2016,Cannon2019}.

\begin{algorithm}[H]
    \caption{Markov chain $\mathcal{M}_{\text{AGG}}$ for aggregation and dispersion in SOPS}
    \label{alg:m_agg}
    \begin{algorithmic}[1]
        \Statex Beginning at any configuration of $N$ particles in a bounded region, fix $\lambda > 1$ and repeat:
        \State Choose a particle $P$ uniformly at random; let $\ell$ be the lattice node it occupies.
        \State Choose an adjacent lattice node $\ell'$ and $q \in (0,1)$ each uniformly at random.
        \If {$\ell'$ is empty and $q < \lambda^{-n}$, where $n$ is the number of neighbors $P$ has at $\ell$} \label{alg:m_agg:prob}
            \State $P$ moves to $\ell'$.
        \Else {} $P$ remains at $\ell$.
        \EndIf
    \end{algorithmic}
\end{algorithm}

Recall that Theorem~\ref{thm:aggregation} analyzes the stationary distribution $\pi$ of the Markov chain $\mathcal{M}_{\text{AGG}}$ for aggregation and dispersion.
In particular, Theorem~\ref{thm:aggregation} was shown in \cite{Cannon2019} to hold for $\pi(\sigma) \propto \lambda^{-b(\sigma)} \propto \lambda^{E(\sigma)}$, where $b(\sigma)$ is the number of ``boundary edges'' of the lattice that have exactly one endpoint occupied by a particle.
So it remains to show that $\mathcal{M}_{\text{AGG}}$ converges to this stationary distribution $\pi$.

\begin{lemma} \label{lem:statdist}
    The unique stationary distribution of $\mathcal{M}_{\text{AGG}}$ is $\pi(\sigma) = \lambda^{-b(\sigma)}/Z$, where $Z = \sum_{\tau}\lambda^{-b(\tau)}$ is a normalizing constant.
\end{lemma}
\begin{proof}
    Let $\sigma$ and $\tau$ be any two SOPS configurations with $\sigma \neq \tau$ such that $\Pr(\sigma, \tau) > 0$, implying that $\tau$ can be reached from $\sigma$ by a single move of some particle $P$.
    Suppose $P$ has $n$ neighbors in $\sigma$ and has $n'$ in $\tau$.
    We must show the \textit{detailed balance condition} holds with respect to the transition probabilities:
    \[\Pr(\sigma, \tau)\pi(\sigma) = \Pr(\tau, \sigma)\pi(\tau)\]
    The algorithms in~\cite{Cannon2016,Cannon2019} were designed using the  Metropolis--Hastings algorithm~\cite{Hastings1970} which specifies transition probabilities $\Pr(\sigma, \tau) = \min\{\pi(\tau)/\pi(\sigma), 1\}$ to capture the ratio between stationary weights of the current and proposed configurations.
    So we have that $\pi(\tau) / \pi(\sigma) = \lambda^{n' - n}$.
    It is then easy to see that this ratio is unchanged by the modified transition probabilities where $\Pr(\sigma, \tau) = \lambda^{-n}$ and $\Pr(\tau, \sigma) = \lambda^{-n'}$, and thus detailed balance is satisfied:
    \[\frac{\Pr(\sigma, \tau)}{\Pr(\tau, \sigma)} = \frac{\lambda^{-n}}{\lambda^{-n'}} = \lambda^{n' - n} = \frac{\pi(\tau)}{\pi(\sigma)}\]
    Therefore, since $\pi$ satisfies detailed balance and $\mathcal{M}_{\text{AGG}}$ is an ergodic finite Markov chain, we conclude that $\pi$ is the unique stationary distribution of $\mathcal{M}_{\text{AGG}}$.
\end{proof}

We conclude by outlining the proof of Theorem~\ref{thm:aggregation} that shows $\mathcal{M}_{\text{AGG}}$ achieves aggregation when $\lambda$ is large enough and dispersion when $\lambda$ is close to one.
Our proof is a series of information-theoretic arguments about the stationary distribution $\pi$.
We use ideas similar to Peierls arguments, which are often used in statistical physics to study phase changes in behavior space for infinite systems~\cite{Friedli2017-statmechlattice}.
In~\cite{Cannon2019} it was shown that, for finite systems, particles of two different colors could either \textit{separate} into monochromatic clusters or \textit{integrate}, indifferent to color.
This separation algorithm can be applied to the setting where a bounded region of the lattice is completely filled with particles that move by ``swapping'' places with their neighbors.
By viewing particles of one color as ``empty space'' and particles of the other color as our particles of interest, the swap moves in the separation algorithm correspond to particle moves within a bounded area.
These are precisely the moves used in our aggregation algorithm, where  separation corresponds to aggregation and integration corresponds to dispersion.
Thus, it is straightforward to leverage the arguments for separation and integration in~\cite{Cannon2019} to show aggregation and dispersion in a bounded region.

For large enough bias $\lambda$, we prove aggregation occurs with high probability as follows.
Using techniques introduced in~\cite{Miracle2011}, we define a map from any configuration without an aggregate to a configuration with an aggregate by (\textit{i}) choosing some scattered particles in a systematic way and (\textit{ii}) rearranging them as an aggregate in a carefully chosen location.
We then show that no aggregate configuration has too many preimages under this map because of the careful way we remove scattered particles.
On the other hand, we show that applying this map to a dispersed configuration leads to a large increase in its stationary probability.
Provided $\lambda$ is large enough that the probability gain outweighs the number of preimages, these two facts imply that aggregated configurations are much more likely to occur in the stationary distribution than dispersed ones.
More formally, the above argument shows that the stationary probability of being in a dispersed configuration is at most $(c_1/\lambda)^{c_2\sqrt{N}}$, where $c_1, c_2 > 0$ are constants that depend on the map described above.
Thus, provided $\lambda$ is large enough, this probability of being in a dispersed configuration is very small, proving that aggregation is achieved with high probability.

When the bias $\lambda$ is close to one, we can prove that dispersion occurs with high probability.
We show that there exist polynomially many events such that if aggregation occurs, then at least one of these events must also occur.
These events correspond to certain regularly-shaped subregions of the lattice being almost entirely occupied by particles.
We then use a Chernoff-type bound to show that each of these events is exponentially unlikely when $\lambda$ is close to one.
This implies that the stationary probability for aggregated configurations is at most the sum of polynomially many terms that are each exponentially small, so dispersion must occur with high probability for this range of $\lambda$.

\subsection*{BOBbot design}

The BOBbot mechanical design was developed in SolidWorks, and its skeleton was 3D printed in ABS plastic by a Stratsys UPrint SE Plus printer at a layer resolution of 0.010 inches and sparse density (\figtext~\ref{fig:bobbotdesign}).
Each BOBbot contains a lithium ion polymer battery (Adafruit Industries) that is equipped with Qi wireless charging for recharging between experiments (Adafruit Industries).
The brushbot design is implemented using an ERM (BestTong) for vibrations and two Pienoy dog toothbrush heads as feet, yielding noisy circular trajectories (Movie S2).
The BOBbot's motor circuitry was assembled on a Printed Circuit Board (PCB) designed in EagleCAD (\figtext~\ref{fig:bobbotcircuitry}).
The PCBs were printed at the Georgia Tech Interdisciplinary Design Commons makerspace and outsourced from JLCPCB.
This circuitry is switched and modulated by a phototransistor (Adafruit Industries), which acts as a proportional controller for motor speed.
Grade N42 neodymium magnets (K\&J Magnetics) are housed in the BOBbot chassis for inter-robot attraction, and can be swapped for magnets of different strengths to modulate the BOBbots' cohesion.
A complete list of BOBbot components can be found in Table~\ref{tab:components}.

To achieve stress sensing, each BOBbot is equipped with four triggers that mechanically deform and close the circuit upon collisions to sense the locally exerted stress (\figtext~\ref{fig:stress}B,~\ref{fig:bobbotcircuitry}).
These triggers are positioned radially in front of the permanent magnets in the chassis.
The stress sensors function such that a robot decreases its motor speed for sufficiently large stress (\figtext~\ref{fig:stress}C).
The analog circuit is designed to reduce the motor's current in a manner proportional to the total number of contacts, starting with roughly $70\%$ reduction for a single triggered sensor (\figtext~\ref{fig:stress}C, top).
When multiple sensors are triggered, a BOBbot's speed is practically negligible.

\subsection*{Simulations}

To simulate the SOPS, we execute the algorithm on a hexagonal lattice.
The size of the lattice in \figtext~\ref{fig:SOPS} is chosen to be sufficiently large so that boundary effects are mitigated.
The size of the lattice for \figtext~\ref{fig:alganalysis} is chosen to match the area density and the number of agents in the physical evolution and algorithm.
To determine the constant $k_0$ in the aggregation metric $AGG_{MC} = N_{MC} / (k_0 P_{MC} \sqrt{N})$, we consider a hexagon with area $N_{MC} = \frac{\sqrt{3}}{4}\ell^2 \cdot 6$ and perimeter $P_{MC} = 6\ell$, setting $k_0$ so that $AGG_{MC} = 1$.
This yields $k_0 = 1/\sqrt{8\sqrt{3}}$.

Beyond the information described in the main text, the DEM simulations faithfully represent the spherical loose magnets with exponentially decaying force housed in each BOBbot's chassis slots, resulting in patchy magnetic interaction as the magnets move freely in their slots.
Attraction between two simulated BOBbots is calculated based on these magnetic spheres' strength and the minimum physical separation between any interacting pair, which depends on the relative position and orientation of the two BOBbots.

To calibrate our DEM simulations, we measure the BOBbots' physical parameters and use these values for the simulated BOBbots (Table~\ref{tab:parameters}).
Most parameters such as the mass and dimensions of each BOBbot are directly measured.
For others, we use a series of experiments designed to isolate individual parameters.
For instance, to avoid possible system errors such as in-plane friction when measuring the magnetic force, we measured the minimum force needed to overwhelm the magnetic force in vertical direction (\figtext~\ref{fig:FM0measure}).
Other indirect measurements involve the translational and rotational drag (\figtext~\ref{fig:dragmeasure},~\ref{fig:rotdragmeasure}).
The key ingredient in these experiments is to use a known force (Earth's gravity) to calibrate these intricate forces.
Details can be found in Section~S2 of the Supplementary Materials.

\begin{table}[ht]
    \centering
    \begin{tabular}{|clll|}
        \hline
        & \textbf{Description} & \textbf{Experiment} & \textbf{Simulation} \\
        \hline\hline
        $m$ & BOBbot mass & 0.060 kg & 0.060 kg \\ 
        \hline
        $R_0$ & BOBbot radius & 0.030 m & 0.030 m \\
        \hline
        $I$ & BOBbot moment of inertia & 2.7e-5 kg$\cdot$m$^2$ & 2.7e-5 kg$\cdot$m$^2$\\
        \hline
        $R_C$ & \makecell[l]{radius of the regular\\ circular motion} & 25 $\pm$ 5 mm & 25 mm\\
        \hline
        $R_{B_0}$ & radius of the magnetic bead & 2.3 mm & 2.0 mm \\
        \hline
        $R_S$ & \makecell[l]{thickness of the magnet\\ cavity shell} & 2.0 mm & 2.0 mm \\
        \hline
        $R_B$ & \makecell[l]{effective radius of\\ the magnetic bead} & 4.3 mm & 4.0 mm\\
        \hline
        $v_0$ & Saturated speed & 48.4 $\pm$ 20.2 mm/s & 60.0 mm/s\\
        \hline
        $\omega_0$ & \makecell[l]{saturated angular velocity\\ of the orbit} & 1.94 $\pm$ 0.81 rad/s & 2.40 rad/s\\
        \hline
        $F_D$ & translational drive & 0.07 N & 0.06 N \\
        \hline
        $\tau_D$ & rotational drive (torque) & 5e-4 N$\cdot$m & 5.5e-4 N$\cdot$m \\
        \hline
        $\eta$ & translational drag coefficient & $\sim$1 kg/s & 1.0 kg/s\\
        \hline
        $\eta_{\varphi}$ & \makecell[l]{rotational drag coefficient} & $\le$3e-4 N$\cdot$m$\cdot$s & 2.3e-4 N$\cdot$m$\cdot$s\\
        \hline
        $F_{M0}$ & magnetic force on contact & 3-35 gf & 3-35 gf\\
        \hline
        $d_0$ & magnetic force decay length & 1.5 mm & 1.5 mm\\
        \hline
        $\mu$ & bot-bot friction coefficient & 0.143 & 0.143 \\
        \hline
        $\mu_W$ & bot-wall friction coefficient &  & 0.143 \\
        \hline
    \end{tabular}
    \caption{\textbf{List of parameters used in physical simulations.}}
    \label{tab:parameters}
\end{table}

The DEM simulations use the Euler-Maruyama method with a time step of 1 ms to integrate the following Newton equations:
\begin{align*}
    m\ddot{\vec{r}} &= F_D\hat{u} - \eta\dot{\vec{r}} + \vec{F}_{\text{env}}(\vec{r},\varphi) + \vec{\xi}(t) \\
    I\ddot{\varphi} &= \tau_D - \eta_{\varphi}\dot{\varphi} + \tau_{\text{env}}(\vec{r},\varphi) + \xi_{\varphi}(t)
\end{align*}
As the agents are in the overdamped regime where $|m\ddot{\vec{r}}| \ll \eta|\dot{\vec{r}}|$, the Newton equations are equivalent to the Langevin equations for active Brownian particles by taking the limit $m,I \to 0$.
\begin{align*}
    \dot{\vec{r}} &= v_0\hat{u} + \vec{F}_{\text{env}}(\vec{r},\varphi)/\eta + \vec{\xi}(t)/\eta \\
    \dot{\varphi} &= \omega_0 + \tau_{\text{env}}(\vec{r},\varphi)/\eta_{\varphi} + \xi_{\varphi}(t)/\eta_{\varphi}
\end{align*}
As we see from the reduced equations, in the steady state, a BOBbot will perform a circular motion with a saturated speed $v_0 = F_D/\eta$ and a frequency of $\omega_0=\tau_D/\eta_{\varphi}$.
This suggests that we can control a BOBbot's speed $v_0$ by changing its motor vibration strength, varying $F_D$.

The initial placement of the BOBbots is achieved by greedy rejection sampling, sequentially placing BOBbots in random positions that do not overlap with the previously placed BOBbots.
A cell list search method is used to speed up the simulation's computation by subdividing the simulated arena into square cells so that, when integrating forces for a given BOBbot, only consider interactions with BOBbots from the same or adjacent cells.
The size of the cells is chosen such that the relative error caused by this approximation is within $10^{-3}$.

\section*{Acknowledgments}

Our BOBbots and their behavior were inspired by granular materials pioneer Prof.\ Robert Behringer (1948--2018) and by discussions originating at the 2018 Granular Matter Gordon Research Conference.
We would like to thank Kurt Weisenfeld, Bulbul Chakraborty, Yasemin Ozkan Aydin, Jennifer Rieser, and Andrew Zangwill for helpful discussions.
We also thank undergraduates Rida Abbas and Lewis Campbell and high school student Paul Aidan Loughlin for their assistance in conducting experiments.
This research was performed in part using resources provided by the Open Science Grid~\cite{pordes2007open,sfiligoi2009pilot}, which is supported by the NSF under award PHY-1148698 and by the U.S. Department of Energy's Office of Science.
We are immensely grateful to Sudarshan Ghonge and Carrie Brown for their help in setting up the Open Science Grid cluster implementations of our DEM simulations.
We also would like to thank our reviewers who contributed insightful ideas that improved this work, including the introduction of disaggregating in the object transport simulations and the connections between our theoretical model and the Cahn--Hilliard equation.
\textbf{Funding:} This work was supported by the Department of Defense under MURI award \#W911NF-19-1-0233 and by NSF awards DMS-1803325 (S.C.), CCF-1422603, CCF-1637393, and CCF-1733680 (A.W.R.), CCF-1637031 and CCF-1733812 (D.R. and D.I.G.), and CCF-1526900 (D.R.).
\textbf{Author Contributions:} B.D., S.L., E.A., and  D.I.G.\ were responsible for experiments, S.L., B.D., R.A., and D.I.G.\ were responsible for simulations, and S.C., J.J.D., A.W.R., and D.R.\ were responsible for the algorithmic analysis.
D.I.G.\ and D.R.\ coordinated the integration of experimental and theoretical perspectives.
\textbf{Data and Materials Availability:} All data needed to evaluate the conclusions in the paper are present in the paper and/or the Supplementary Materials.
Additional data related to this paper may be requested from the authors.

\bibliographystyle{unsrt} 
\bibliography{ref}

\begin{thebibliography}{10}

\bibitem{Magurran1990}
Anne~E. Magurran.
\newblock The adaptive significance of schooling as an anti-predator defence in
  fish.
\newblock {\em Annales Zoologici Fennici}, 27(2):51--66, 1990.

\bibitem{Mlot2011}
Nathan~J. Mlot, Craig~A. Tovey, and David~L. Hu.
\newblock Fire ants self-assemble into waterproof rafts to survive floods.
\newblock {\em Proceedings of the National Academy of Sciences},
  108(19):7669--7673, 2011.

\bibitem{Liu2015}
Jintao Liu, Arthur Prindle, Jacqueline Humphries, Mar\c{c}al Gabalda-Sagarra,
  Munehiro Asally, Dong-Yeon~D. Lee, San Ly, Jordi Garcia-Ojalvo, and
  G\"urol~M. S\"uel.
\newblock Metabolic co-dependence gives rise to collective oscillations within
  biofilms.
\newblock {\em Nature}, 523(7562):550--554, 2015.

\bibitem{Brambilla2013}
Manuele Brambilla, Eliseo Ferrante, Mauro Birattari, and Marco Dorigo.
\newblock Swarm robotics: a review from the swarm engineering perspective.
\newblock {\em Swarm Intelligence}, 7(1):1--41, 2013.

\bibitem{Bayindir2016}
Levent Bayindir.
\newblock A review of swarm robotics tasks.
\newblock {\em Neurocomputing}, 172:292--321, 2016.

\bibitem{Dorigo2020-futureswarms}
Marco Dorigo, Guy Theraulaz, and Vito Trianni.
\newblock Reflections on the future of swarm robotics.
\newblock {\em Science Robotics}, 5(49):eabe4385, 2020.

\bibitem{Elamvazhuthi2019-meanfield}
Karthik Elamvazhuthi and Spring Berman.
\newblock Mean-field models in swarm robotics: a survey.
\newblock {\em Bioinspiration {\&} Biomimetics}, 15(1):015001, 2019.

\bibitem{Flocchini2019}
Paola Flocchini, Giuseppe Prencipe, and Nicola Santoro, editors.
\newblock {\em Distributed Computing by Mobile Entities}.
\newblock Springer International Publishing, Switzerland, 2019.

\bibitem{Mayya2019}
Siddharth Mayya, Gennaro Notomista, Dylan Shell, Seth Hutchinson, and Magnus
  Egerstedt.
\newblock Non-uniform robot densities in vibration driven swarms using phase
  separation theory.
\newblock In {\em 2019 IEEE/RSJ International Conference on Intelligent Robots
  and Systems}, pages 4106--4112, 2019.

\bibitem{Notomista2019}
Gennaro Notomista, Siddharth Mayya, Anirban Mazumdar, Seth Hutchinson, and
  Magnus Egerstedt.
\newblock A study of a class of vibration-driven robots: Modeling, analysis,
  control and design of the brushbot.
\newblock In {\em 2019 IEEE/RSJ International Conference on Intelligent Robots
  and Systems}, pages 5101--5106, 2019.

\bibitem{Hamann2018-swarmrobotics}
Heiko Hamann.
\newblock {\em Swarm Robotics: A Formal Approach}.
\newblock Springer, 2018.

\bibitem{Hines2017-softactuators}
Lindsey Hines, Kirstin Petersen, Guo~Zhan Lum, and Metin Sitti.
\newblock Soft actuators for small-scale robotics.
\newblock {\em Advanced Materials}, 29(13):1603483, 2017.

\bibitem{Xie2019}
Hui Xie, Mengmeng Sun, Xinjian Fan, Zhihua Lin, Weinan Chen, Lei Wang, Lixin
  Dong, and Qiang He.
\newblock Reconfigurable magnetic microrobot swarm: Multimode transformation,
  locomotion, and manipulation.
\newblock {\em Science Robotics}, 4(28):eaav8006, 2019.

\bibitem{Wolpert2019}
David~H. Wolpert.
\newblock The stochastic thermodynamics of computation.
\newblock {\em Journal of Physics A: Mathematical and Theoretical},
  52(19):193001, 2019.

\bibitem{Rubenstein2014}
Michael Rubenstein, Alejandro Cornejo, and Radhika Nagpal.
\newblock Programmable self-assembly in a thousand-robot swarm.
\newblock {\em Science}, 345(6198):795--799, 2014.

\bibitem{Piranda2018}
Benoit Piranda and Julien Bourgeois.
\newblock Designing a quasi-spherical module for a huge modular robot to create
  programmable matter.
\newblock {\em Autonomous Robots}, 42(8):1619--1633, 2018.

\bibitem{Fates2011}
Nazim Fat\`es and Nikolaos Vlassopoulos.
\newblock A robust aggregation method for quasi-blind robots in an active
  environment.
\newblock In {\em ICSI 2011}, 2011.

\bibitem{Gauci2014}
Melvin Gauci, Jianing Chen, Wei Li, Tony~J. Dodd, and Roderich Gro\ss.
\newblock Self-organized aggregation without computation.
\newblock {\em International Journal of Robotics Research}, 33(8):1145--1161,
  2014.

\bibitem{Ozdemir2019}
Anil \"Ozdemir, Melvin Gauci, Andreas Kolling, Matthew~D. Hall, and Roderich
  Gro\ss.
\newblock Spatial coverage without computation.
\newblock In {\em 2019 IEEE International Conference on Robotics and
  Automation}, ICRA 2019, pages 1346--1353. IEEE, 2019.

\bibitem{Garnier2009}
Simon Garnier, Jacques Gautrais, Masoud Asadpour, Christian Jost, and Guy
  Theraulaz.
\newblock Self-organized aggregation triggers collective decision making in a
  group of cockroach-like robots.
\newblock {\em Adaptive Behavior}, 17(2):109--133, 2009.

\bibitem{Correll2011}
Nikolaus Correll and Alcherio Martinoli.
\newblock Modeling and designing self-organized aggregation in a swarm of
  miniature robots.
\newblock {\em The International Journal of Robotics Research}, 30(5):615--626,
  2011.

\bibitem{Li2019}
Shuguang Li, Richa Batra, David Brown, Hyun-Dong Chang, Nikhil Ranganathan,
  Chuck Hoberman, Daniela Rus, and Hod Lipson.
\newblock Particle robotics based on statistical mechanics of loosely coupled
  components.
\newblock {\em Nature}, 567:361--365, 2019.

\bibitem{Agrawal2017-tunablestructures}
Mayank Agrawal, Isaac~R. Bruss, and Sharon~C. Glotzer.
\newblock Tunable emergent structures and traveling waves in mixtures of
  passive and contact-triggered-active particles.
\newblock {\em Soft Matter}, 13(37):6332--6339, 2017.

\bibitem{Deblais2018-boundarycontrol}
A.~Deblais, T.~Barois, T.~Guerin, P.~H. Delville, R.~Vaudaine, J.~S.
  Lintuvuori, J.~F. Boudet, J.~C. Baret, and H.~Kellay.
\newblock Boundaries control collective dynamics of inertial self-propelled
  robots.
\newblock {\em Phys. Rev. Lett.}, 120:188002, 2018.

\bibitem{Savoie2019-smarticleensemble}
William Savoie, Thomas~A. Berrueta, Zachary Jackson, Ana Pervan, Ross
  Warkentin, Shengkai Li, Todd~D. Murphey, Kurt Wiesenfeld, and Daniel~I.
  Goldman.
\newblock A robot made of robots: Emergent transport and control of a smarticle
  ensemble.
\newblock {\em Science Robotics}, 4(34):eaax4316, 2019.

\bibitem{Cannon2016}
Sarah Cannon, Joshua~J. Daymude, Dana Randall, and Andr\'ea~W. Richa.
\newblock A {M}arkov chain algorithm for compression in self-organizing
  particle systems.
\newblock In {\em Proceedings of the 2016 ACM Symposium on Principles of
  Distributed Computing}, PODC '16, pages 279--288, New York, NY, USA, 2016.
  ACM.

\bibitem{Metropolis1953}
Nicholas Metropolis, Arianna~W. Rosenbluth, Marshall~N. Rosenbluth, Augusta~H.
  Teller, and Edward Teller.
\newblock Equation of state calculations by fast computing machines.
\newblock {\em Journal of Chemical Physics}, 21:1087--1092, 1953.

\bibitem{Hastings1970}
W.~K. Hastings.
\newblock Monte carlo sampling methods using {Markov} chains and their
  applications.
\newblock {\em Biometrika}, 57:97--109, 1970.

\bibitem{Cannon2019}
Sarah Cannon, Joshua~J. Daymude, Cem G{\"o}kmen, Dana Randall, and
  Andr{\'e}a~W. Richa.
\newblock A local stochastic algorithm for separation in heterogeneous
  self-organizing particle systems.
\newblock In {\em Approximation, Randomization, and Combinatorial Optimization.
  Algorithms and Techniques (APPROX/RANDOM 2019)}, pages 54:1--54:22, 2019.

\bibitem{andreotti2013granular}
Bruno Andreotti, Yo{\"e}l Forterre, and Olivier Pouliquen.
\newblock {\em Granular media: between fluid and solid}.
\newblock Cambridge University Press, 2013.

\bibitem{lim2019cluster}
Melody~X Lim, Anton Souslov, Vincenzo Vitelli, and Heinrich~M Jaeger.
\newblock Cluster formation by acoustic forces and active fluctuations in
  levitated granular matter.
\newblock {\em Nature Physics}, 15(5):460--464, 2019.

\bibitem{melo1995hexagons}
Francisco Melo, Paul~B Umbanhowar, and Harry~L Swinney.
\newblock Hexagons, kinks, and disorder in oscillated granular layers.
\newblock {\em Physical review letters}, 75(21):3838, 1995.

\bibitem{eshuis2007phase}
Peter Eshuis, Ko~Van Der~Weele, Devaraj Van Der~Meer, Robert Bos, and Detlef
  Lohse.
\newblock Phase diagram of vertically shaken granular matter.
\newblock {\em Physics of Fluids}, 19(12):123301, 2007.

\bibitem{keys2007measurement}
Aaron~S Keys, Adam~R Abate, Sharon~C Glotzer, and Douglas~J Durian.
\newblock Measurement of growing dynamical length scales and prediction of the
  jamming transition in a granular material.
\newblock {\em Nature physics}, 3(4):260--264, 2007.

\bibitem{goldman2006signatures}
Daniel~I Goldman and Harry~L Swinney.
\newblock Signatures of glass formation in a fluidized bed of hard spheres.
\newblock {\em Physical review letters}, 96(14):145702, 2006.

\bibitem{rericha2001shocks}
Erin~C Rericha, Chris Bizon, Mark~D Shattuck, and Harry~L Swinney.
\newblock Shocks in supersonic sand.
\newblock {\em Physical review letters}, 88(1):014302, 2001.

\bibitem{howell1999stress}
Daniel Howell, RP~Behringer, and Christian Veje.
\newblock Stress fluctuations in a 2d granular couette experiment: a continuous
  transition.
\newblock {\em Physical Review Letters}, 82(26):5241, 1999.

\bibitem{corwin2005structural}
Eric~I Corwin, Heinrich~M Jaeger, and Sidney~R Nagel.
\newblock Structural signature of jamming in granular media.
\newblock {\em Nature}, 435(7045):1075--1078, 2005.

\bibitem{bi2011jamming}
Dapeng Bi, Jie Zhang, Bulbul Chakraborty, and Robert~P Behringer.
\newblock Jamming by shear.
\newblock {\em Nature}, 480(7377):355--358, 2011.

\bibitem{Mitarai2006-wetgranular}
Namiko Mitarai and Franco Nori.
\newblock Wet granular materials.
\newblock {\em Advances in Physics}, 55(1-2):1--45, 2006.

\bibitem{Hemmerle2016-cohesivegranular}
Arnaud Hemmerle, Matthias Schröter, and Lucas Goehring.
\newblock A cohesive granular material with tunable elasticity.
\newblock {\em Scientific Reports}, 6(1):35630, 2016.

\bibitem{kummel2013circular}
Felix K{\"u}mmel, Borge ten Hagen, Raphael Wittkowski, Ivo Buttinoni, Ralf
  Eichhorn, Giovanni Volpe, Hartmut L{\"o}wen, and Clemens Bechinger.
\newblock Circular motion of asymmetric self-propelling particles.
\newblock {\em Physical review letters}, 110(19):198302, 2013.

\bibitem{Jahanshahi2017}
Soudeh Jahanshahi, Hartmut L{\"o}wen, and Borge Ten~Hagen.
\newblock Brownian motion of a circle swimmer in a harmonic trap.
\newblock {\em Physical Review E}, 95(2):022606, 2017.

\bibitem{Ramaswamy2017}
Sriram Ramaswamy.
\newblock Active matter.
\newblock {\em Journal of Statistical Mechanics: Theory and Experiment},
  2017:054002, 2017.

\bibitem{jerrum1993polynomial}
Mark Jerrum and Alistair Sinclair.
\newblock Polynomial-time approximation algorithms for the ising model.
\newblock {\em SIAM Journal on computing}, 22(5):1087--1116, 1993.

\bibitem{randall1999sampling}
Dana Randall and David Wilson.
\newblock Sampling spin configurations of an {I}sing system.
\newblock In {\em Proceedings of the Tenth Annual ACM-SIAM Symposium on
  Discrete Algorithms}, SODA '99, pages 959–--960, 1999.

\bibitem{penrose1991mean}
O\_ Penrose.
\newblock A mean-field equation of motion for the dynamic {I}sing model.
\newblock {\em Journal of Statistical Physics}, 63(5-6):975--986, 1991.

\bibitem{Bayindir2009-aggregation}
Levent Bayindir and Erol Sahin.
\newblock Modeling self-organized aggregation in swarm robotic systems.
\newblock In {\em 2009 IEEE Swarm Intelligence Symposium}, pages 1--8, 2009.

\bibitem{Valentini2016-collectiveperception}
Gabriele Valentini, Davide Brambilla, Heiko Hamann, and Marco Dorigo.
\newblock Collective perception of environmental features in a robot swarm.
\newblock In {\em International Conference on Swarm Intelligence}, ANTS 2016,
  pages 65--76, 2016.

\bibitem{majmudar2007jamming}
TS~Majmudar, M~Sperl, Stefan Luding, and Robert~P Behringer.
\newblock Jamming transition in granular systems.
\newblock {\em Physical review letters}, 98(5):058001, 2007.

\bibitem{aguilar2018collective}
J~Aguilar, D~Monaenkova, V~Linevich, W~Savoie, B~Dutta, H-S Kuan, MD~Betterton,
  MAD Goodisman, and DI~Goldman.
\newblock Collective clog control: Optimizing traffic flow in confined
  biological and robophysical excavation.
\newblock {\em Science}, 361(6403):672--677, 2018.

\bibitem{Wilson2014-stochastictransport}
Sean Wilson, Theodore~P. Pavlic, Ganesh~P. Kumar, Aurélie Buffin, Stephen~C.
  Pratt, and Spring Berman.
\newblock Design of ant-inspired stochastic control policies for collective
  transport by robotic swarms.
\newblock {\em Swarm Intelligence}, 8:303--327, 2014.

\bibitem{Feinerman2018}
Ofer Feinerman, Itai Pinkoviezky, Aviram Gelblum, Ehud Fonio, and Nir~S. Gov.
\newblock The physics of cooperative transport in groups of ants.
\newblock {\em Nature Physics}, 14:683--693, 2018.

\bibitem{Stenhammar2013}
Joakim Stenhammar, Adriano Tiribocchi, Rosalind~J. Allen, Davide Marenduzzo,
  and Michael~E. Cates.
\newblock Continuum theory of phase separation kinetics for active brownian
  particles.
\newblock {\em Phys. Rev. Lett.}, 111:145702, Oct 2013.

\bibitem{Friedli2017-statmechlattice}
Sacha Friedli and Yvan Velenik.
\newblock {\em Statistical Mechanics of Lattice Systems: A Concrete
  Mathematical Introduction}.
\newblock Cambridge University Press, Cambridge, 2017.

\bibitem{Miracle2011}
Sarah Miracle, Dana Randall, and Amanda~Pascoe Streib.
\newblock Clustering in interfering binary mixtures.
\newblock In {\em Approximation, Randomization, and Combinatorial Optimization.
  Algorithms and Techniques}, APPROX '11, RANDOM '11, pages 652--663, 2011.

\bibitem{pordes2007open}
Ruth Pordes, Don Petravick, Bill Kramer, Doug Olson, Miron Livny, Alain Roy,
  Paul Avery, Kent Blackburn, Torre Wenaus, Frank W{\"u}rthwein, et~al.
\newblock The open science grid.
\newblock {\em Journal of Physics: Conference Series}, 78:012057, 2007.

\bibitem{sfiligoi2009pilot}
Igor Sfiligoi, Daniel~C Bradley, Burt Holzman, Parag Mhashilkar, Sanjay Padhi,
  and Frank Wurthwein.
\newblock The pilot way to grid resources using glideinwms.
\newblock In {\em 2009 WRI World congress on computer science and information
  engineering}, volume~2, pages 428--432. IEEE, 2009.

\bibitem{toral1995large-sm}
Ra{\'u}l Toral, Amitabha Chakrabarti, and James~D Gunton.
\newblock Large scale simulations of the two-dimensional cahn-hilliard model.
\newblock {\em Physica A: Statistical Mechanics and its Applications},
  213(1-2):41--49, 1995.

\bibitem{chen2019positivity-sm}
Wenbin Chen, Cheng Wang, Xiaoming Wang, and Steven~M Wise.
\newblock Positivity-preserving, energy stable numerical schemes for the
  cahn-hilliard equation with logarithmic potential.
\newblock {\em Journal of Computational Physics: X}, 3:100031, 2019.

\end{thebibliography}

\clearpage

\renewcommand{\thefigure}{S\arabic{figure}}
\setcounter{figure}{0}
\renewcommand{\thetable}{S\arabic{table}}
\setcounter{table}{0}

\section*{Supplementary Materials}

\figtext~\ref{fig:bobbotdesign}. Cross-sectional views of the BOBbot mechanical design.\\
\figtext~\ref{fig:bobbotcircuitry}. BOBbot circuitry.\\
\figtext~\ref{fig:expplatform}. Experimental platform design and details.\\
\figtext~\ref{fig:FM0measure}. Calibration experiment for calculating magnet force $F_{M0}$.\\
\figtext~\ref{fig:dragmeasure}. Calibration experiment for calculating translational drag coefficient $\eta$.\\
\figtext~\ref{fig:rotdragmeasure}. Calibration experiment for calculating rotational drag coefficient $\eta_{\varphi}$.\\
\figtext~\ref{fig:airflowmeasure}. Boundary airflow effects in experiment and simulation.\\
\figtext~\ref{fig:NMCvsV}. Dependence of maximum cluster size $N_{MC}$ on BOBbot speed $v_0$ and curvature $R_C$.\\
\figtext~\ref{fig:detachprob}. Probability of detachment for various magnetic attraction.\\
\figtext~\ref{fig:transporttrajectories}. Object transport trajectories.\\
\figtext~\ref{fig:transporttheory}. Toy model for object transport.\\
\figtext~\ref{fig:shuffling}. Transport enhanced by disaggregating.\\
\figtext~\ref{fig:IsingSOPS}. Examples showing $\Delta H_\text{Ising} = 2J\Delta H_\text{SOPS}$.\\
\figtext~\ref{fig:bifurcation}. Critical surface tension $\gamma$ and bias parameter $\lambda$.\\
\figtext~\ref{fig:CHsim}. Pattern formation below and above critical $\lambda$.\\
\figtext~\ref{fig:thermoDynLim}. Approaching $P_{MC} \propto N_{MC}^{1/2}$ with periodic boundary conditions.\\
Table~\ref{tab:components}. List of BOBbot components.\\
Movie S1. Aggregation dynamics in a self-organizing particle system (SOPS).\\
Movie S2. Individual BOBbot dynamics.\\
Movie S3. Aggregation and dispersion in BOBbot collectives.\\
Movie S4. Aggregation and dispersion in simulated BOBbot collectives.\\
Movie S5. BOBbot collective dynamics in large simulated systems.\\
Movie S6. Enhanced BOBbot aggregation using mechanical stress sensing.\\
Movie S7. Object transport by BOBbot collectives.

\subsection*{S1. BOBbot design and manufacturing}

\figtext~\ref{fig:bobbotdesign} depicts various cross-sectional views of a BOBbot's design and corresponding skeletal structure.
\figtext~\ref{fig:bobbotcircuitry} shows the PCB design and assembly.
Table~\ref{tab:components} lists all components used in BOBbot manufacturing.
Finally, \figtext~\ref{fig:expplatform} shows the design and details of the experimental platform.

\begin{figure}[ht]
    \centering
    \includegraphics[width=0.55\textwidth]{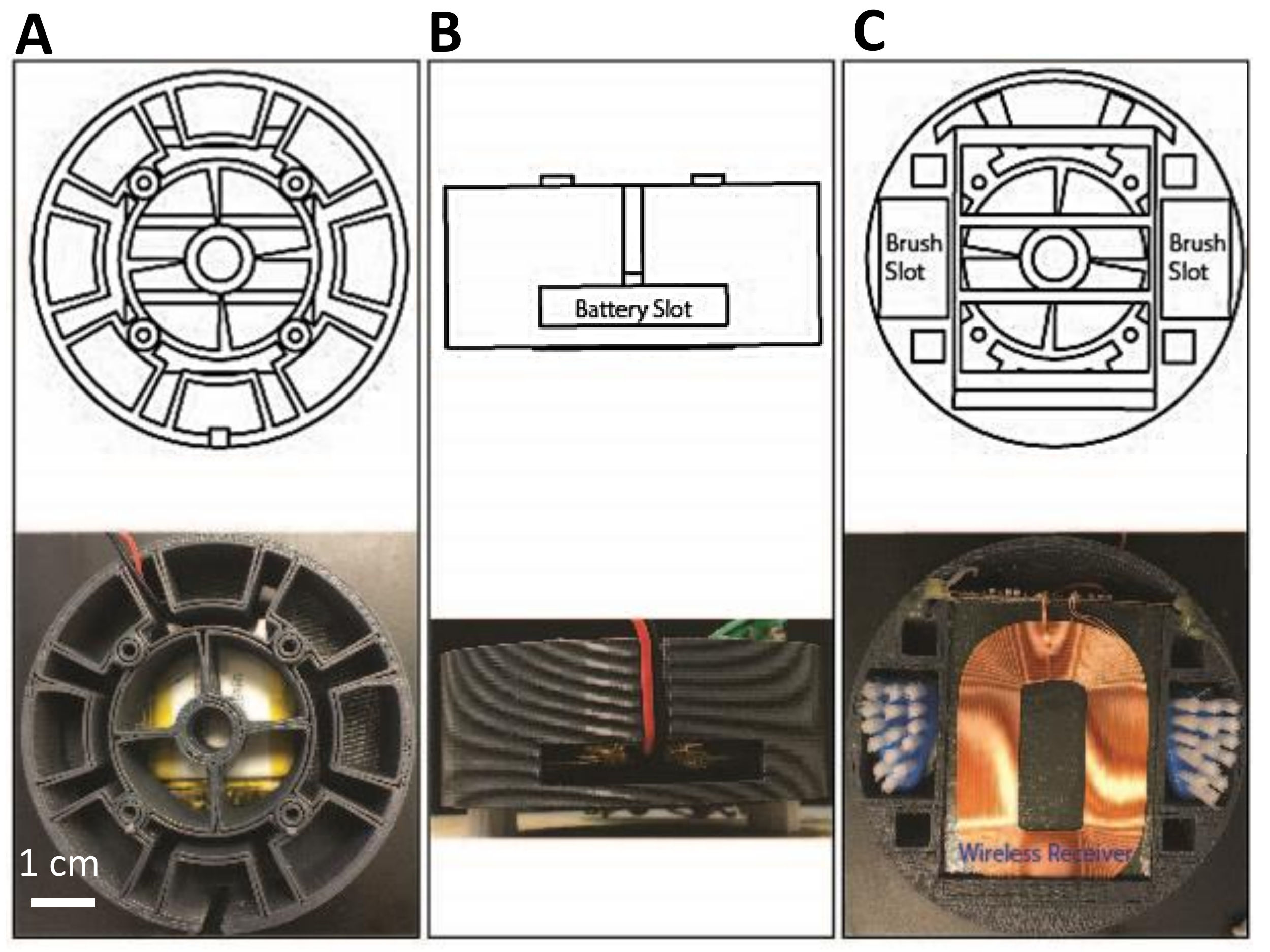}
    \caption{\textbf{Cross-sectional views of the BOBbot mechanical design.}
    SolidWorks designs and assembled versions of (\textbf{A}) the BOBbot shell and magnet slots, (\textbf{B}) the battery slot, and (\textbf{C}) the brush slots and wireless QR charge receiver.}
    \label{fig:bobbotdesign}
\end{figure}

\begin{figure}[ht]
    \centering
    \includegraphics[width=0.55\textwidth]{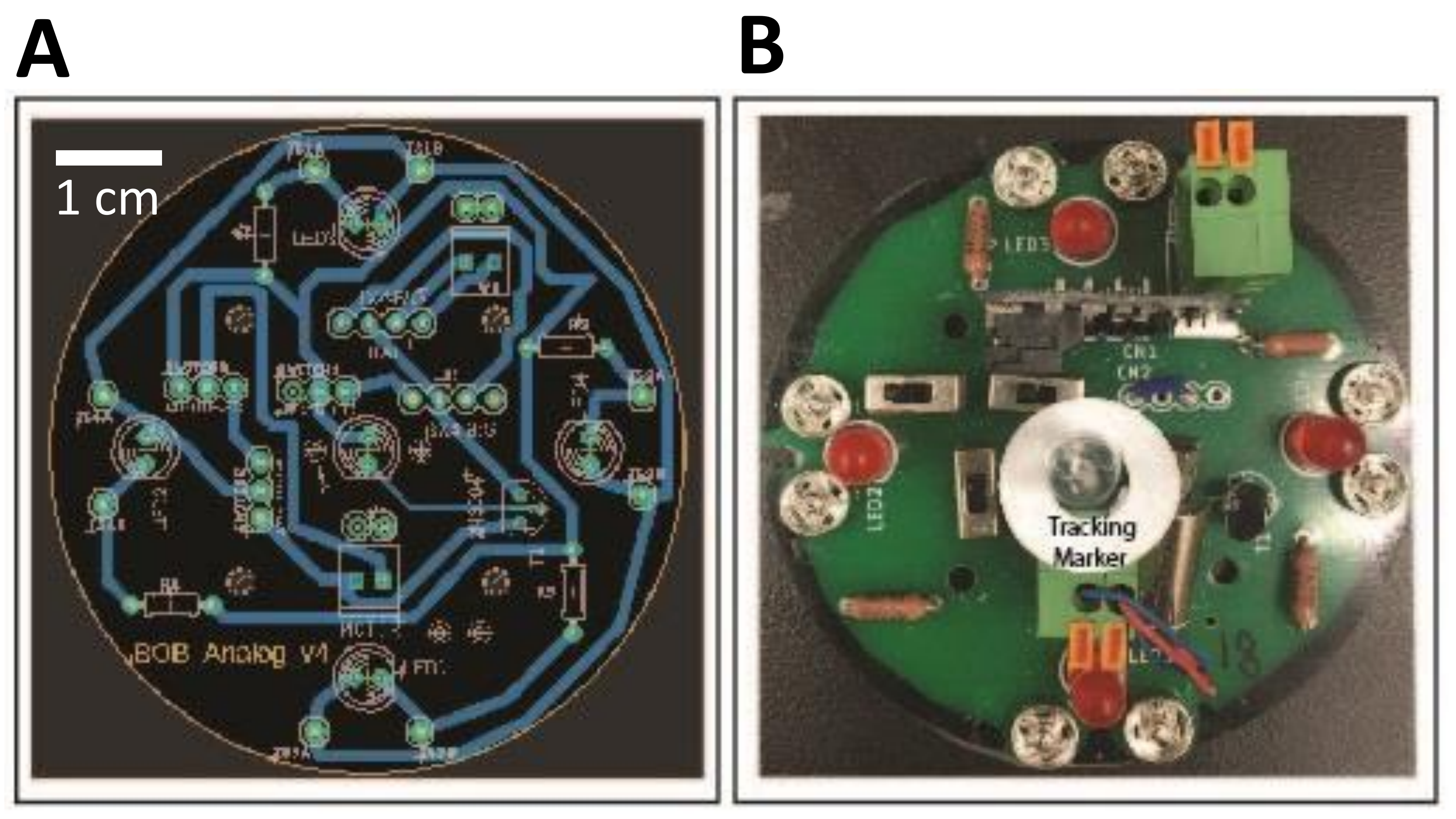}
    \caption{\textbf{BOBbot circuitry.}
    (\textbf{A}) The analog PCB design, made in EagleCAD.
    (\textbf{B}) The printed and assembled PCB.}
    \label{fig:bobbotcircuitry}
\end{figure}

\begin{table}[ht]
    \centering
    \footnotesize
    \begin{tabular}{|l|l|l|}
        \hline
        \textbf{Component} & \textbf{Manufacturer} & \textbf{Product Name} \\
        \hline
        ERM Motor & BestTong & \href{https://www.amazon.com/dp/B073JKQ9LN}{DC 3.7V 9500RPM Vibrating Coreless Brushed Motor} \\
        Brushes & Pienoy & \href{https://www.amazon.com/dp/B07GQPVFNB}{Double-Headed Pet Toothbrush} \\
        Magnets & K\&J Magnetics & \href{https://www.kjmagnetics.com/products.asp?cat=12}{S2 and S3} \\
        Battery & Adafruit Industries & \href{https://www.adafruit.com/product/1578}{Lithium Ion Polymer Battery 3.7V 500mAh} \\
        Battery Module & Adafruit Industries & \href{https://www.adafruit.com/product/1904}{Micro-Lipo Charger} \\
        Qi Transmitter & Adafruit Industries & \href{https://www.adafruit.com/product/2162}{Universal Qi Wireless Charging Transmitter} \\
        Qi Receiver & Adafruit Industries & \href{https://www.adafruit.com/product/1901}{Universal Qi Wireless Receiver Module} \\
        Red/Black Wiring & Adafruit Industries & \href{https://www.adafruit.com/product/288}{Solid-Core Wire Spool} \\
        LED & KingSo & \href{https://www.amazon.com/dp/B07FKKN2TC}{500pcs LED Diode Lights} \\
        Phototransistor & Adafruit Industries & \href{https://www.adafruit.com/product/2831}{Photo Transistor Light Sensor} \\
        Resistors & Vishay/Dale & \href{https://www.mouser.com/ProductDetail/Vishay-Dale/RN55C50R0BB14?qs=sGAEpiMZZMsPqMdJzcrNwns55x5W4xb46fw7oH3zv94\%3D}{Metal Film Resistor 1/10 Watt 50 Ohm 0.1\% 50ppm} \\
        Transistors & ON Semiconductor & \href{https://www.mouser.com/ProductDetail/ON-Semiconductor-Fairchild/2N3904BU?qs=sGAEpiMZZMshyDBzk1\%2FWi4G1GLBZKHK10QNLT596yR0\%3D}{General Purpose Bipolar Transistor} \\
        Switches & Pololu & \href{https://www.pololu.com/product/1408}{Mini Slide Switch 3-Pin, SPDT, 0.3A} \\
        Terminal Block & Pololu & \href{https://www.pololu.com/product/2425}{Screwless Terminal Block: 2-Pin, 0.1'' Pitch} \\
        Button Snaps & Adafruit Industries & \href{https://www.adafruit.com/product/1126}{Sewable Snaps, 5mm Diameter} \\
        Masking Tape & Daigger & \href{https://www.amazon.com/dp/B00BQZ1FRA}{DAI-T34-27-C Assorted Label Tape Pack} \\
        Jumper Cables & Anezus & \href{https://www.amazon.com/gp/B07LG81W22}{700pcs Jumper Wire Kit Breadboard Wires} \\
        \hline
    \end{tabular}
    \normalsize
    \caption{\textbf{List of BOBbot components.}}
    \label{tab:components}
\end{table}

\begin{figure}[ht]
    \centering
    \includegraphics[width=0.8\textwidth]{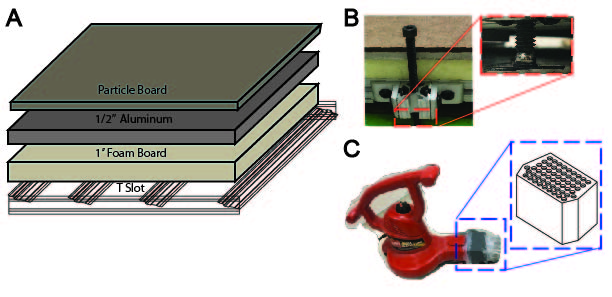}
    \caption{\textbf{Experimental platform design and details.}
    (\textbf{A}) The experimental platform is composed of a T-slot base supporting the foam board, aluminum, and particle board layers.
    (\textbf{B}) Levelling screws in the T-slot framing allow for incline adjustment.
    (\textbf{C}) A leaf blower with a multi-pronged tygon tubing attachment provide airflow to the PVC pipe boundary to mitigate boundary effects.}
    \label{fig:expplatform}
\end{figure}

\newpage

\subsection*{S2. Calibration experiments}

The DEM simulation parameters are calibrated to match the physical BOBbot features.
Many parameters such as the mass and dimensions of each BOBbot are easily measured.
However, other parameters are better calculated by conducting simple experiments.
The first such experiment (\figtext~\ref{fig:FM0measure}) calculates the magnetic force $F_{M0}$ between two magnets when their BOBbots' shells are touching.
The first magnet is placed in a BOBbot shell attached to a rigid stand; a second shell is then tethered beneath the first by placing the second magnet inside it.
Thus, the second shell falls once its weight exceeds $F_{M0}$.
To leverage this insight, a cup is tethered to the second shell and BBs are added to the cup one-by-one until the second shell falls (\figtext~\ref{fig:FM0measure}A).
The weight of the shell, cup, and BBs are then measured to obtain a value of $F_{M0}$ that is precise up to $0.1$ g, the weight of a single BB (\figtext~\ref{fig:FM0measure}B).
On a log-linear plot of force, our measurements show exponential decay, which aligns closely with those reported by the magnet manufacturer (\figtext~\ref{fig:FM0measure}C). A power law fit would gives an exponent of $-18$, which is far off from the $-4$ from dipole-dipole interaction, thus indicating an exponential decay as a better representation.

\begin{figure}[ht]
    \centering
    \includegraphics[width=0.7\textwidth]{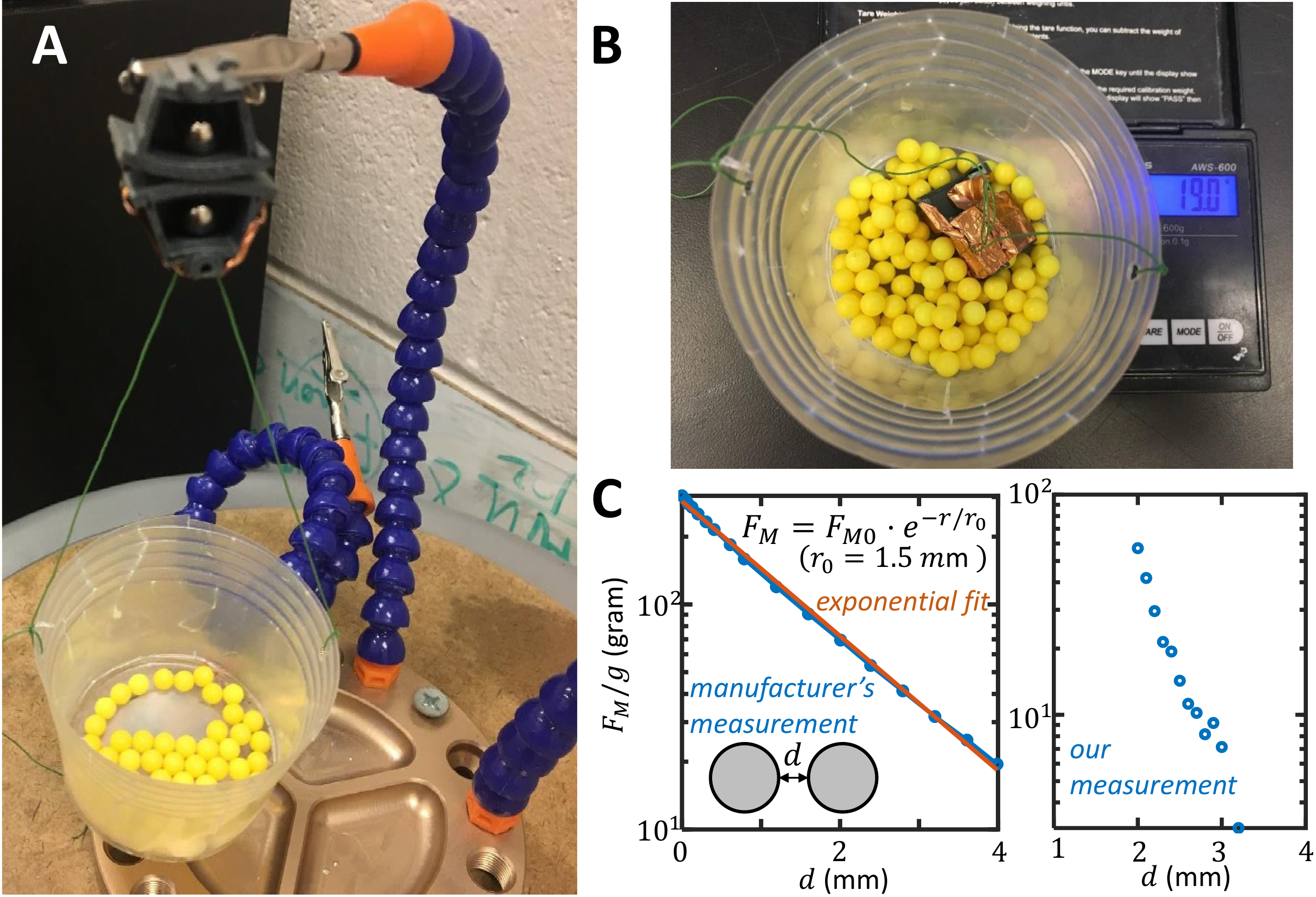}
    \caption{\textbf{Calibration experiment for calculating magnet force $F_{M0}$.}
    (\textbf{A}) The experimental setup for calculating $F_{M0}$.
    (\textbf{B}) Measuring the weight of the tethered apparatus once it falls gives a close approximation of $F_{M0}$.
    (\textbf{C}) Left: the magnetic force's decay with the separation $d$ between two magnetic beads from the manufacturer for $4.6$mm-diameter beads. Right: our measurement for $6.4$ mm-diameter beads.}
    \label{fig:FM0measure}
\end{figure}

Each BOBbot's position $\vec{r}$ and orientation $\varphi$ changes at a constant rate subject to noise.
A BOBbot's constant translational speed $v_0$ comes from the competing driving force $F_D\hat{u}$ and the translational drag $-\eta\dot{\vec{r}}$.
Similarly, each BOBbot's constant rotational speed $\omega_0$ comes from the competing driving torque $\tau_D$ and the rotational drag $-\eta_\varphi \dot{\varphi}$.
The steady-state speeds therefore follow $v_0 = F_D/\eta$ and $\omega_0 = \tau_D/\eta_\varphi$.
We again use simple experiments to determine the drive and drag.
To measure the translational drag $\eta$, we compare a BOBbot's trajectory when it is on a $0^\circ$ incline versus a tilted incline.
In the former, the BOBbot circles regularly with some noise; in the latter, this regular circling is stretched towards the direction of gravity on the incline (\figtext~\ref{fig:dragmeasure}, top).
Using the known gravitational force on the BOBbot, we can calculate the translational drag force and coefficient $\eta$.
We then simulate a BOBbot's motion using different translational drag coefficients; the one that produces the trajectory most closely matching those in the experiments is chosen as the simulation $\eta$ (\figtext~\ref{fig:dragmeasure}).

\begin{SCfigure}[][ht]
    \centering
    \includegraphics[width=0.7\textwidth]{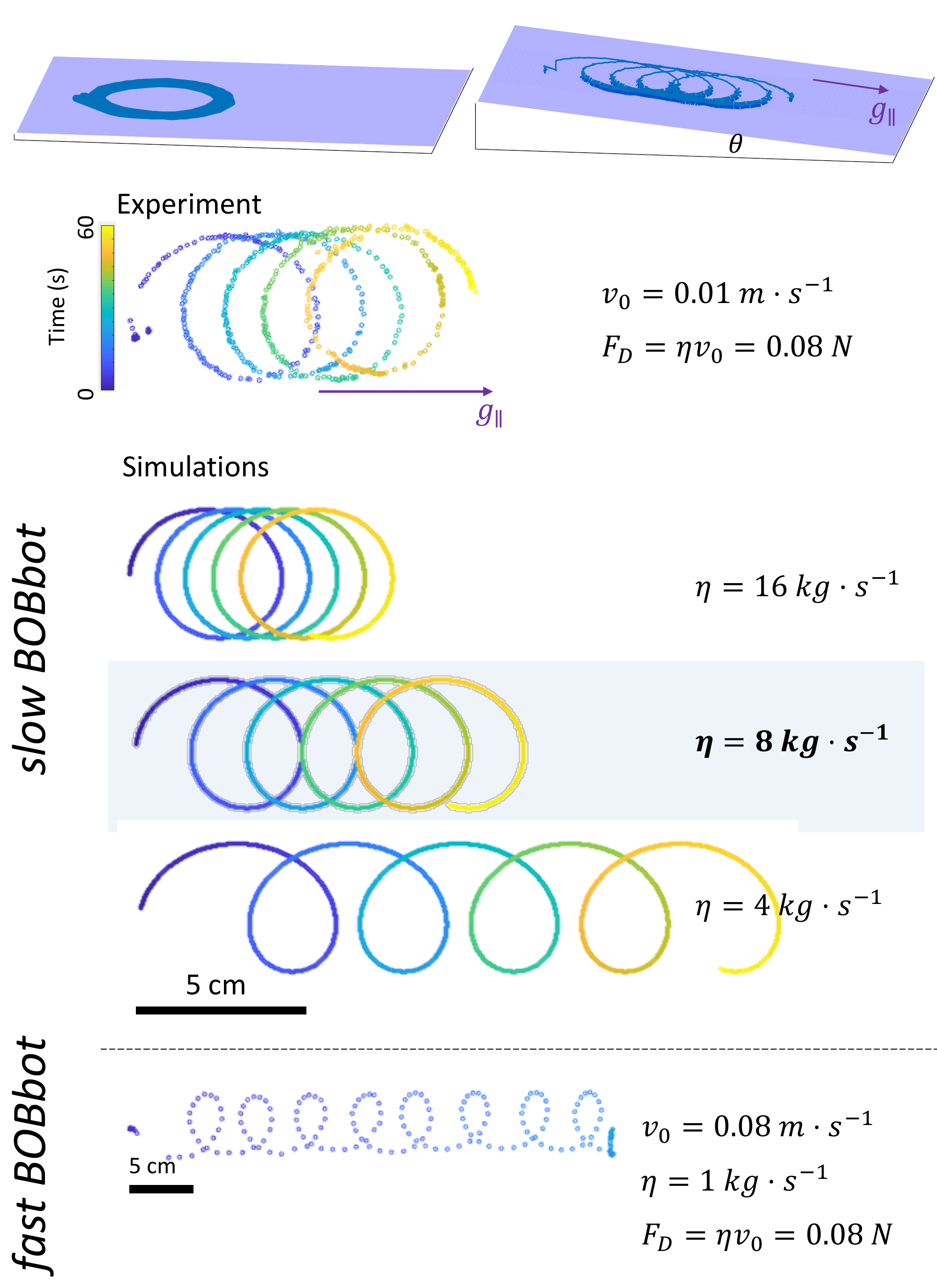}
    \caption{\textbf{Calibration experiment for calculating translational drag coefficient $\eta$.}
    When a BOBbot is driven on a level plane, it circles regularly with some noise.
    When placed on a tilted incline, its trajectory is stretched towards the direction of gravity on the incline.
    Using this known force, we measure the drag force by simulating BOBbot trajectories on a tilted incline using different drag coefficients, comparing each trajectory's stretch to that of the experiment.
    The correct drag produces a close approximate of the experimental trajectory.
    We find that viscosity varies between BOBbots, implying that the their speeds also vary.
    The first three trajectories are from a BOBbot with relatively slow velocity $v_0$; the last is from a fast BOBbot.}
    \label{fig:dragmeasure}
\end{SCfigure}

The measurement of the rotational drag $\eta_\varphi$ exploits its balance with the driving torque.
To measure the rotational torque exerted on a BOBbot, a very light rigid straw is attached across the diameter of a BOBbot (\figtext~\ref{fig:rotdragmeasure}).
We then let the BOBbot use the straw to push objects at various arm lengths.
For a given obstacle to push, the rotational torque is obtained by finding the largest torque of friction on an obstacle to balance.
We decrease the arm length from a large value to a point the BOBbots can just push the obstacle.
Given the measured saturated angular velocity $\omega_0$, the rotational drag can be inferred as $\eta_\varphi = \tau_D/\omega_0$.

\begin{figure}[ht]
    \centering
    \includegraphics[width=0.9\textwidth]{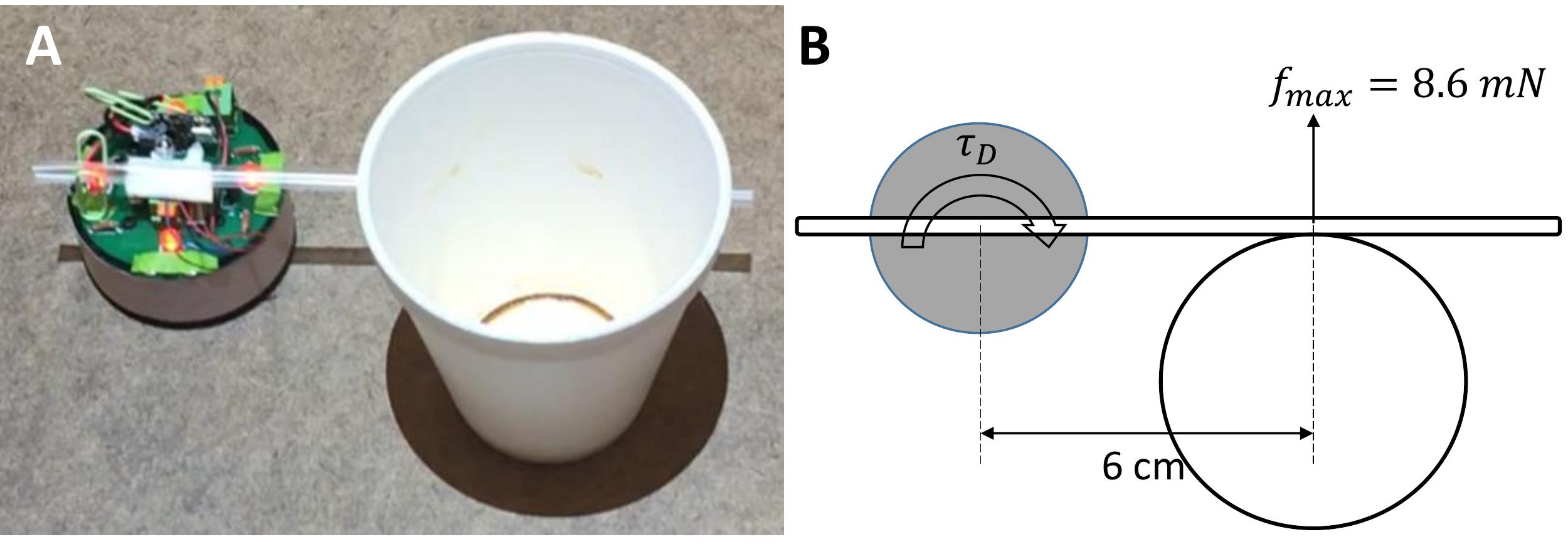}
    \caption{\textbf{Calibration experiment for calculating rotational drag coefficient $\eta_\varphi$.}
    (\textbf{A}) The experimental setup and (\textbf{B}) the corresponding force diagram, where $f_{max}$ denotes the largest frictional torque that the driving torque $\tau_D$ can balance.}
    \label{fig:rotdragmeasure}
\end{figure}

Many of our preliminary experiments were adulterated by boundary effects that caused small groups of BOBbots to collect at the edges and corners of the arena, affecting steady state properties.
We mitigate these affects using airflow-based boundary repulsion.
To characterize these airflow effects, a BOBbot is placed close to the boundary and its trajectory is tracked with and without airflow (\figtext~\ref{fig:airflowmeasure}).
The corresponding simulation parameters are then chosen to match the average characteristics of these experimental trajectories.
The airflow force profile is chosen to match the decay length observed in the example experiment (which is $R_A = 6R_0$).
The resting speed of the bot used in this experiment is $v_0 = 3$~cm/s.
Note that the decay length chosen in the simulation runs throughout our study is $2R_0$ and $v_0 = 6$~cm/s.

\begin{figure}[ht]
    \centering
    \includegraphics[width=0.5\textwidth]{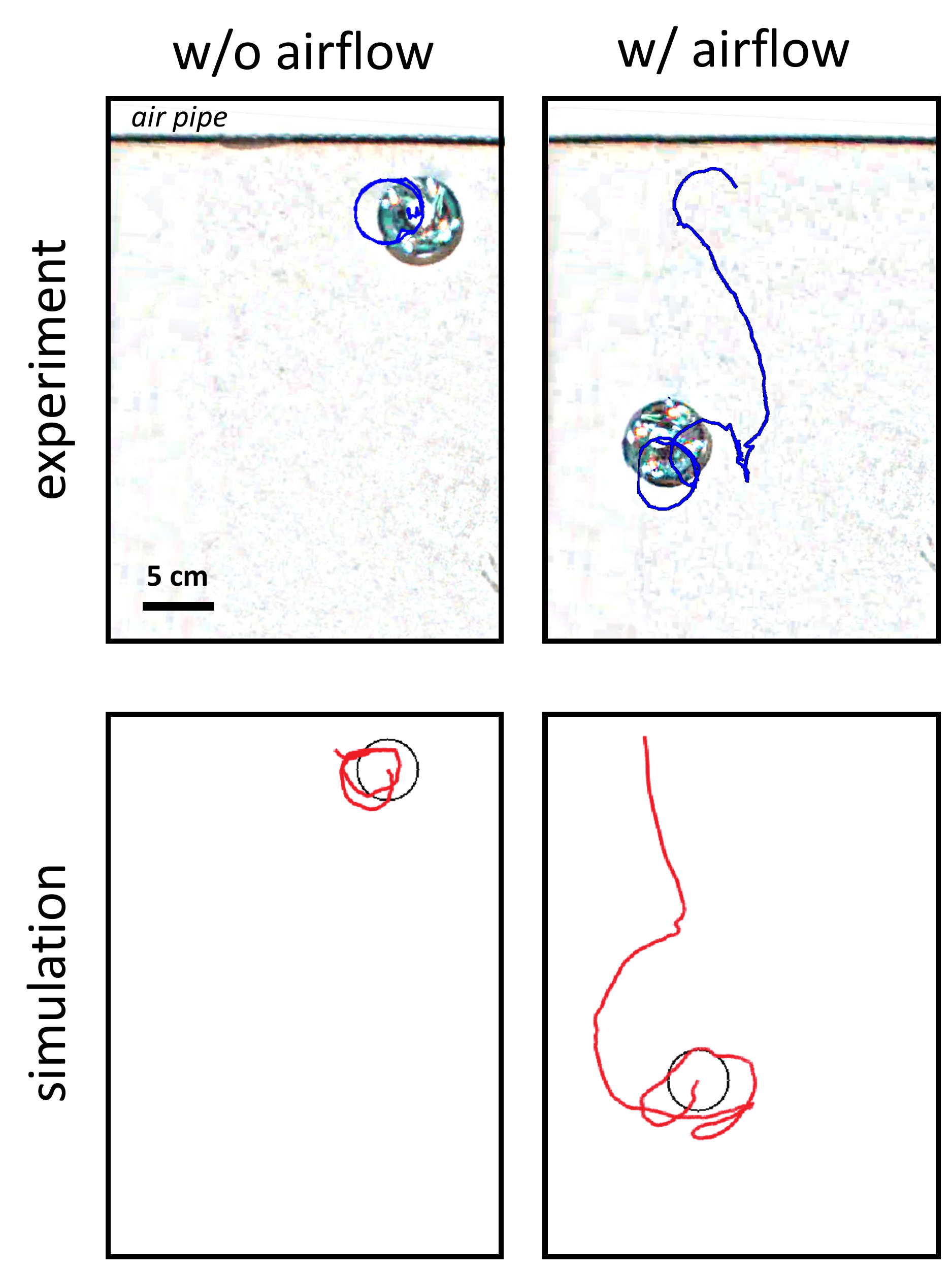}
    \caption{\textbf{Boundary airflow effects in experiment and simulation.}
    Movie S6 shows BOBbot trajectories with and without airflow effects.}
    \label{fig:airflowmeasure}
\end{figure}

\newpage

\subsection*{S3. Sensitivity to initial conditions}

The BOBbots move in deterministic but noisy circular trajectories.
Although the noise only causes small deviations in the individual trajectories, the combination of all the forces in the ensemble makes the system behavior very sensitive to initial conditions.
Simulations with the same random seed started from almost exactly the same initial conditions except for a discrepancy of 0.5 mm of one robot's initial position deviate from each other significantly after 1 minute (\figtext~\ref{fig:sensitivity}).
We therefore regard this system as ergodic, yielding reasonable statistical sampling on longer time scales.

\begin{figure}[ht]
    \centering
    \includegraphics[width=0.9\textwidth]{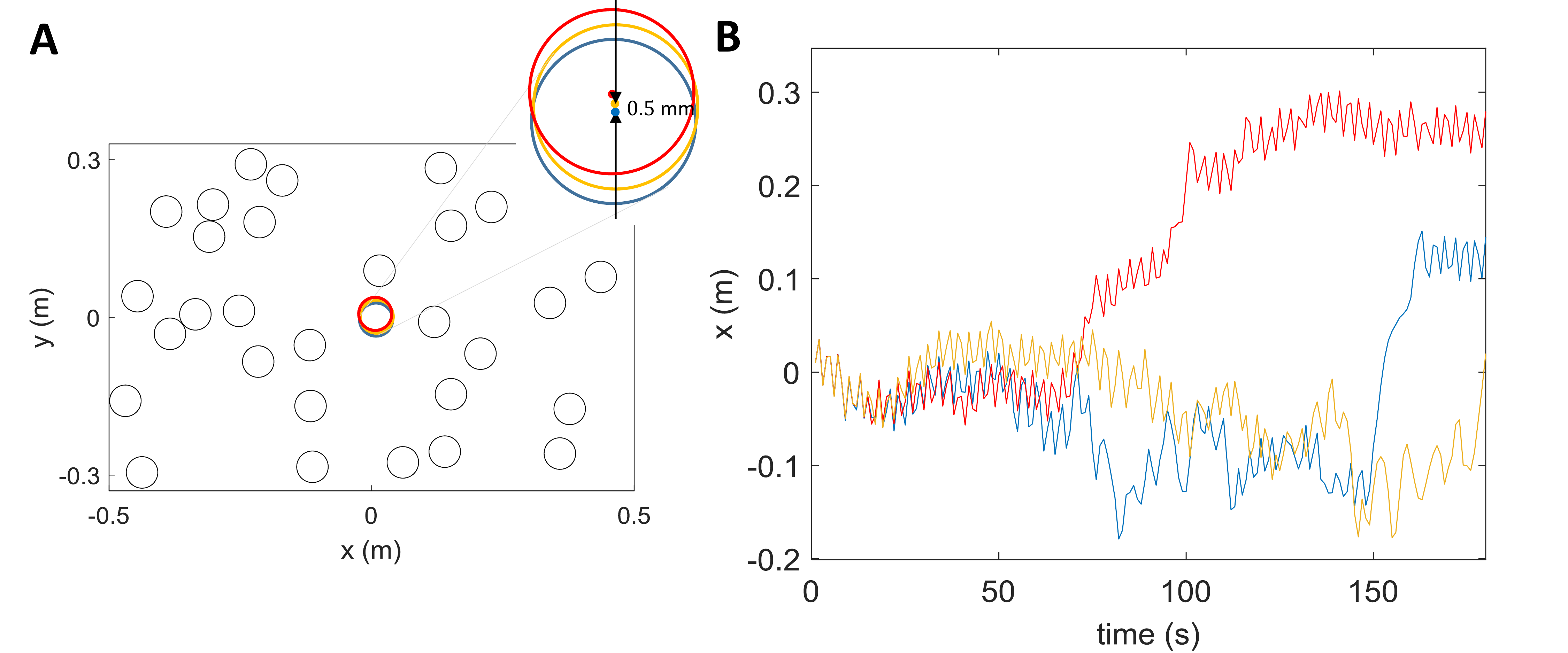}
    \caption{\textbf{Ensemble sensitivity to initial conditions.}
    (\textbf{A}) To measure the ensemble's sensitivity to initial conditions, we simulated three runs with the same random seeds and nearly identical initial conditions, with the exception of a spatial difference of 0.5 mm for the initial position of one BOBbot, shown in blue, yellow, and red.
    (\textbf{B}) The $x$-position over time of the perturbed BOBbot for different starting positions, where the curve color corresponds to the starting position in (A).
    The trials only coincide for roughly the first 30 s, and after one minute they diverge significantly.}
    \label{fig:sensitivity}
\end{figure}

\subsection*{S4. Dependence of maximum cluster size on BOBbot speed and curvature}

To investigate the effect of the BOBbots' individual speeds on the size of the maximum cluster, we run simulations for $F_{M0} = 3, 7, 10$~g with 20 repetitions for 8 speeds equally spaced in range $v_0 = 1$--$8$~cm/s.
\figtext~\ref{fig:NMCvsV} shows how the maximum cluster size $N_{MC}$ decreases as the BOBbots' individual speed is increased.
We also find that $N_{MC}$ increases with larger radii of curvature corresponding to decreased torque.

\begin{figure}[ht]
    \centering
    \includegraphics[width=0.9\textwidth]{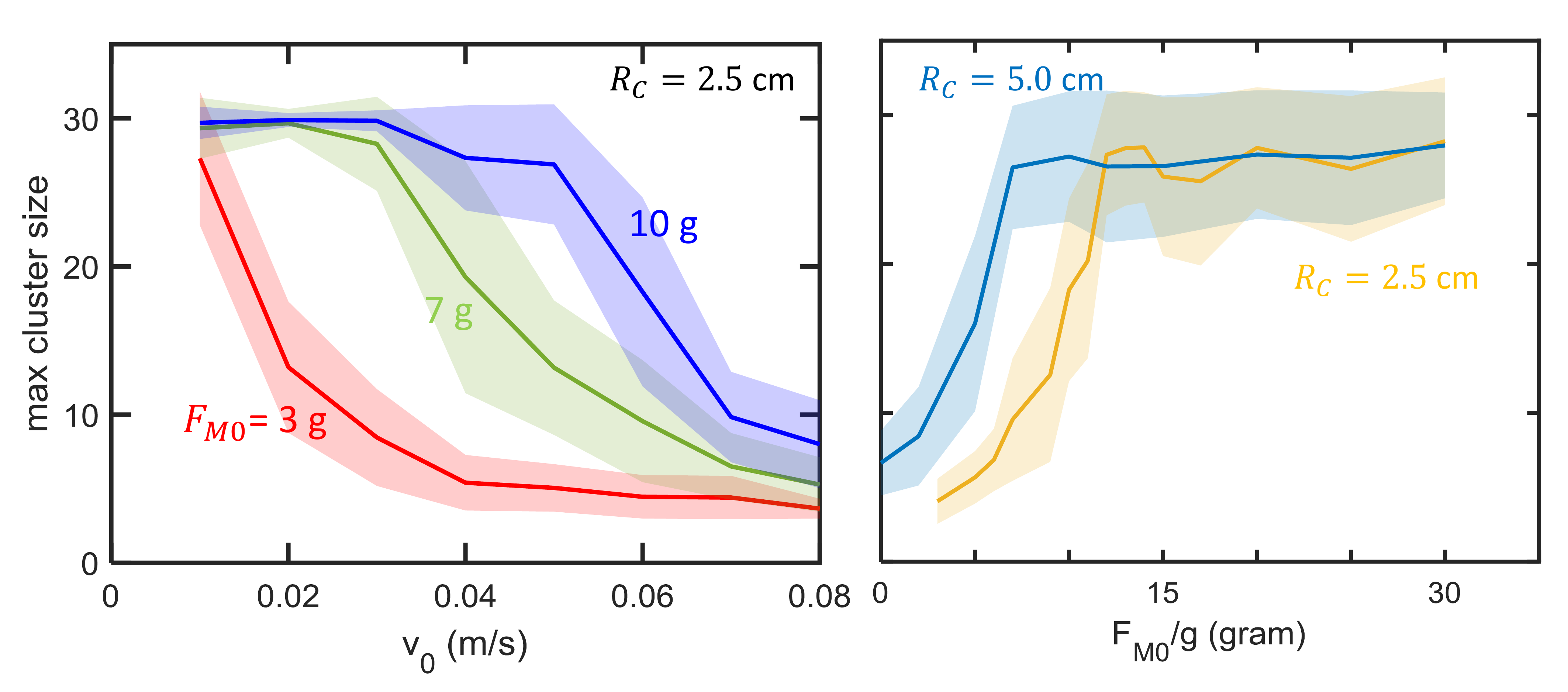}
    \caption{\textbf{Dependence of maximum cluster size $N_{MC}$ on BOBbot speed $v_0$ and curvature $R_C$.}
    The statistics shows the ensemble average from 20 simulations for each data point.}
    \label{fig:NMCvsV}
\end{figure}

\subsection*{S5. Detachment probability}

In \figtext~\ref{fig:alganalysis}A we showed the probability of a BOBbot detaching from three neighbors.
\figtext~\ref{fig:detachprob} shows the complete set of detachment probabilities for various magnetic attraction strengths $F_{M0}$.
We find that the $\lambda_\text{eff}$ increases with the magnetic attraction.
For a particular attraction strength, $\lambda_\text{eff}$ only varies within a small range such that the fitted exponential lines are nearly parallel to each other.

\begin{figure}[ht]
    \centering
    \includegraphics[width=0.7\textwidth]{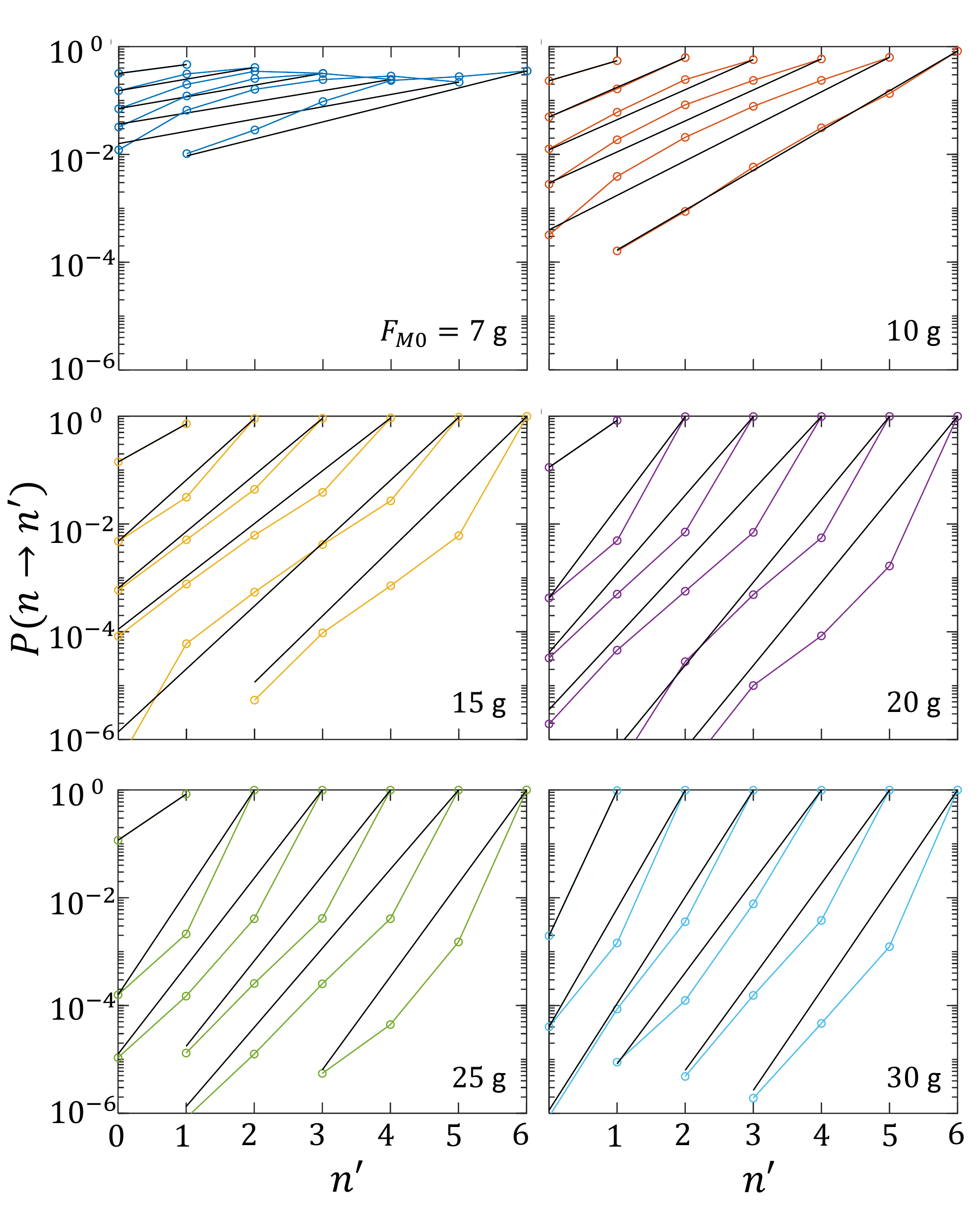}
    \caption{\textbf{Probability of detachment for various magnetic attraction.}
    Colored lines show the simulation data and black lines show the exponential fits.}
    \label{fig:detachprob}
\end{figure}

\subsection*{S6. Object transport by BOBbot collectives}

\begin{figure}[ht]
    \centering
    \includegraphics[width=0.6\textwidth]{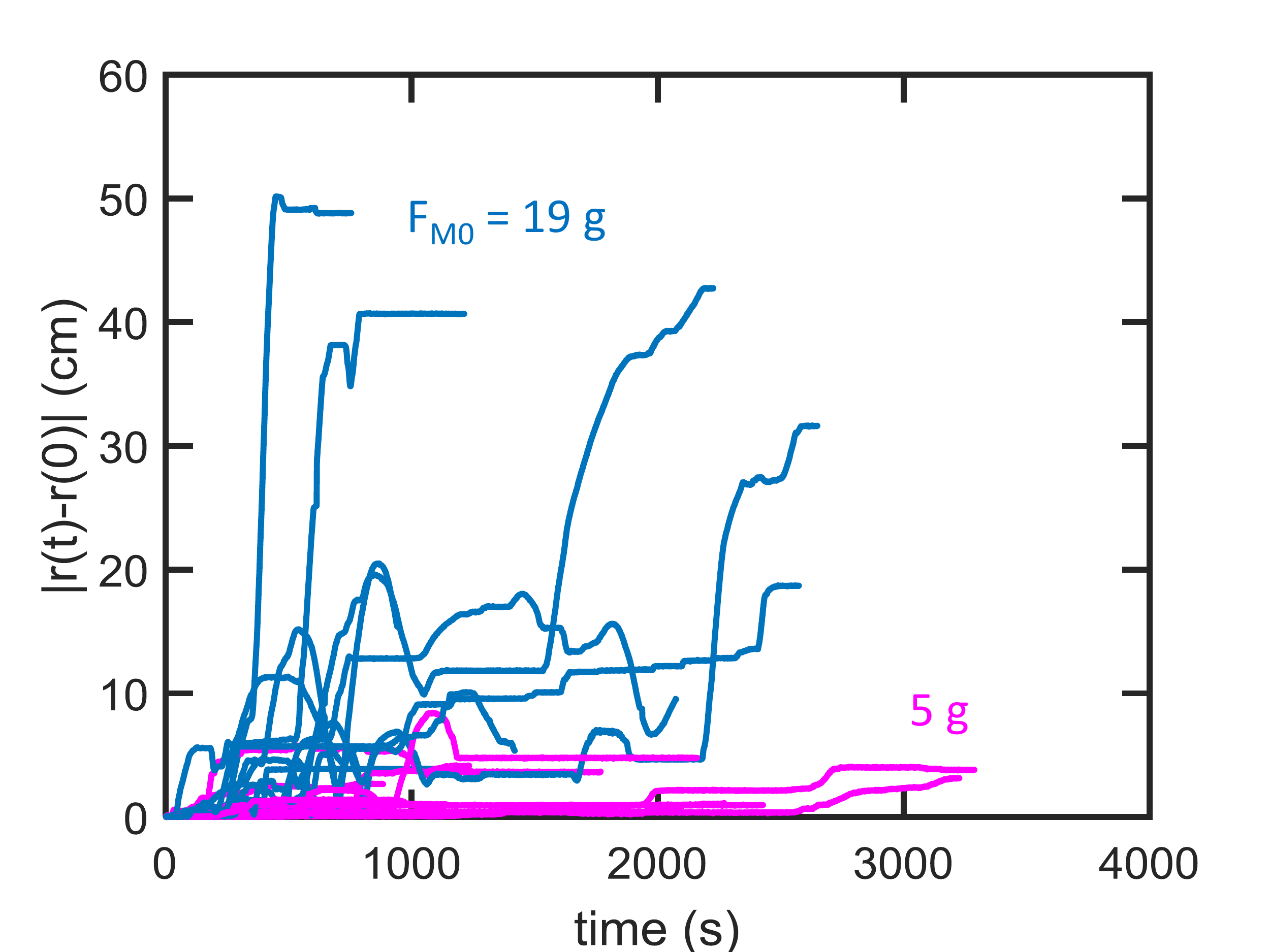}
    \caption{\textbf{Object transport trajectories.} Displacement of the box impurity over time for BOBbot collectives with $F_{M0} = 5$~g (magenta) and $19$~g (blue).}
    \label{fig:transporttrajectories}
\end{figure}

\begin{figure}[ht]
    \centering
    \includegraphics[width=0.5\textwidth]{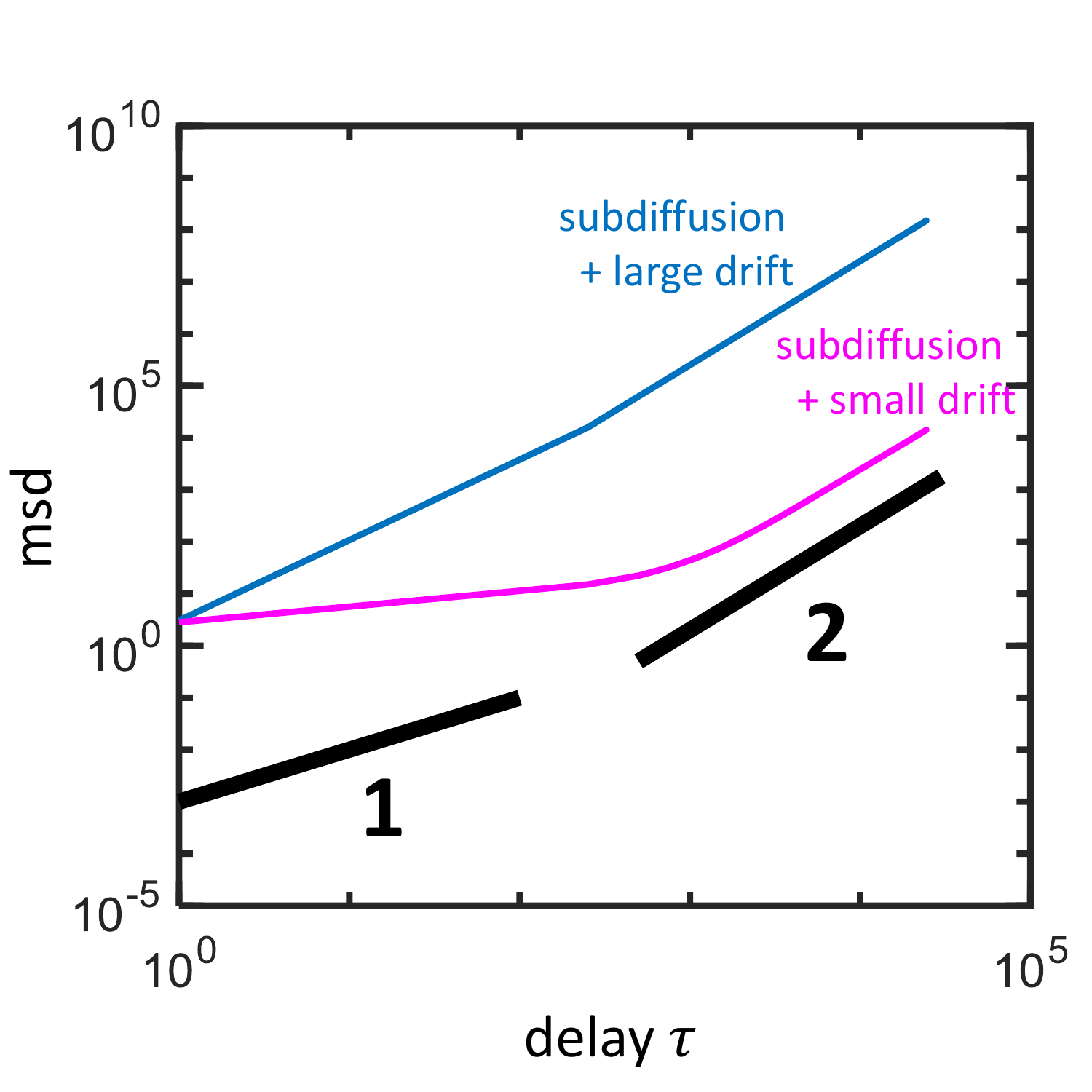}
    \caption{\textbf{Toy model for object transport.} 
    Both trajectories include subdiffusive motion and drift.
    The blue curve has large drift while the magenta curve has small drift, representing transport by the strongly and weakly attractive BOBbot collectives, respectively.
    This model produces a qualitative match with the experiments, demonstrating the origin of the different types of mean-squared displacement over time in the transport experiments.}
    \label{fig:transporttheory}
\end{figure}

\figtext~\ref{fig:transporttrajectories} shows the displacement of the box impurity over time for weakly attractive and strongly attractive BOBbot collectives.
While nearly all experimental runs with strongly attractive collectives exhibit rapid and large displacement, some of the weakly attractive collectives exhibit two-stage transport dynamics that start with very little displacement and eventually achieve larger displacement.
We posit that the weakly attractive collectives' two-stage dynamics are composed of subdiffusion arising from the BOBbots' collisions and a small drift caused by the persistent heterogeneity of the BOBbots around the impurity.
To validate this hypothesis, we developed a toy model in MATLAB where the subdiffusion $\mathbf{r}_H(t)$ with mean-squared displacement $\langle |\mathbf{r}_H(t+\tau)-\mathbf{r}_H(t)|^2\rangle_t=Ct^{2H}$ (for $H < 1/2$) is generated by the fractional Brownian motion generator (from the MATLAB Wavelet Toolbox) and is added to a drift motion $\mathbf{r}_D(t) = \mathbf{v}_D t$.
The relative magnitude difference between the subdiffusion and drift is chosen to match the experiments.
When the drift is small (i.e., $|\mathbf{v}_D| = 0.005$), we observe two-stage transport dynamics consistent with the experiments
(\figtext~\ref{fig:transporttheory}, magenta).
On the other hand, when the drift is dominant over the subdiffusion as in the strongly attractive collectives (i.e., $|\mathbf{v}_D| = 0.5$), the toy model reproduces the nearly ballistic trajectories observed in experiment (\figtext~\ref{fig:transporttheory}, blue).
In fact, the mean-squared displacement of this composed motion $\text{MSD}(\mathbf{r}_H+\mathbf{r}_D)$ is related to the purely subdiffusive $\text{MSD}(\mathbf{r}_H)$ as:
\begin{align*}
    \text{MSD}(\mathbf{r}) &= \langle|\mathbf{r}(t + \tau) - \mathbf{r}(t)|^2\rangle_t\\
    &= \langle|\mathbf{r}_H(t+\tau)+\mathbf{r}_D(t+\tau)-\mathbf{r}_H(t)-\mathbf{r}_D(t)|^2\rangle_t\\
    &= \langle|(\mathbf{r}_H(t + \tau) - \mathbf{r}_H(t)) - \mathbf{v}_D t|^2\rangle_t\\
    &= \langle|(\mathbf{r}_H(t + \tau) - \mathbf{r}_H(t))|^2\rangle_t + \langle 2\mathbf{v}_D \cdot (\mathbf{r}_H(t + \tau) - \mathbf{r}_H(t))\rangle_t + \langle|\mathbf{v}_D|^2 t^2\rangle_t\\
    &= \text{MSD}(\mathbf{r}_H) + 2\mathbf{v}_D \cdot \langle(\mathbf{r}_H(t + \tau) - \mathbf{r}_H(t))\rangle_t + |\mathbf{v}_D|^2 t^2\\
    &= Ct^{2H} + |\mathbf{v}_D|^2 t^2 \label{eq:transportTheo}
\end{align*}
where $\langle(\mathbf{r}_H(t + \tau) - \mathbf{r}_H(t))\rangle_t$ vanishes due to the isotropy of subdiffusion.
The final two equations demonstrate how the subdiffusive power is dominated by the ballistic power $2$ when the drift speed $|\mathbf{v}_D|$ is large.

DEM simulations of the impurity transport task use the same obstacle parameters (e.g., rectangular shape and a mass of $60$ g) as in the experiments and approximates the friction coefficient by $0.25$.
The friction of the box is integrated all over the contact area.
The friction on each point is anti-parallel to the instantaneous velocity.
Sample animations of the simulations can be found in Movie S7.

We additionally performed DEM simulations with intermittent time periods with no attraction in order to dissolve and disaggregate formed aggregates.
Movie S7 demonstrates how aggregates initially missing the obstacle to transport can find it successfully after disaggregating.
All three disaggregating sequences we investigated result in more successful transports to the boundary when compared to the base attractive case without disaggregating.

\begin{figure}[ht]
    \centering
    \includegraphics[width=0.9\textwidth]{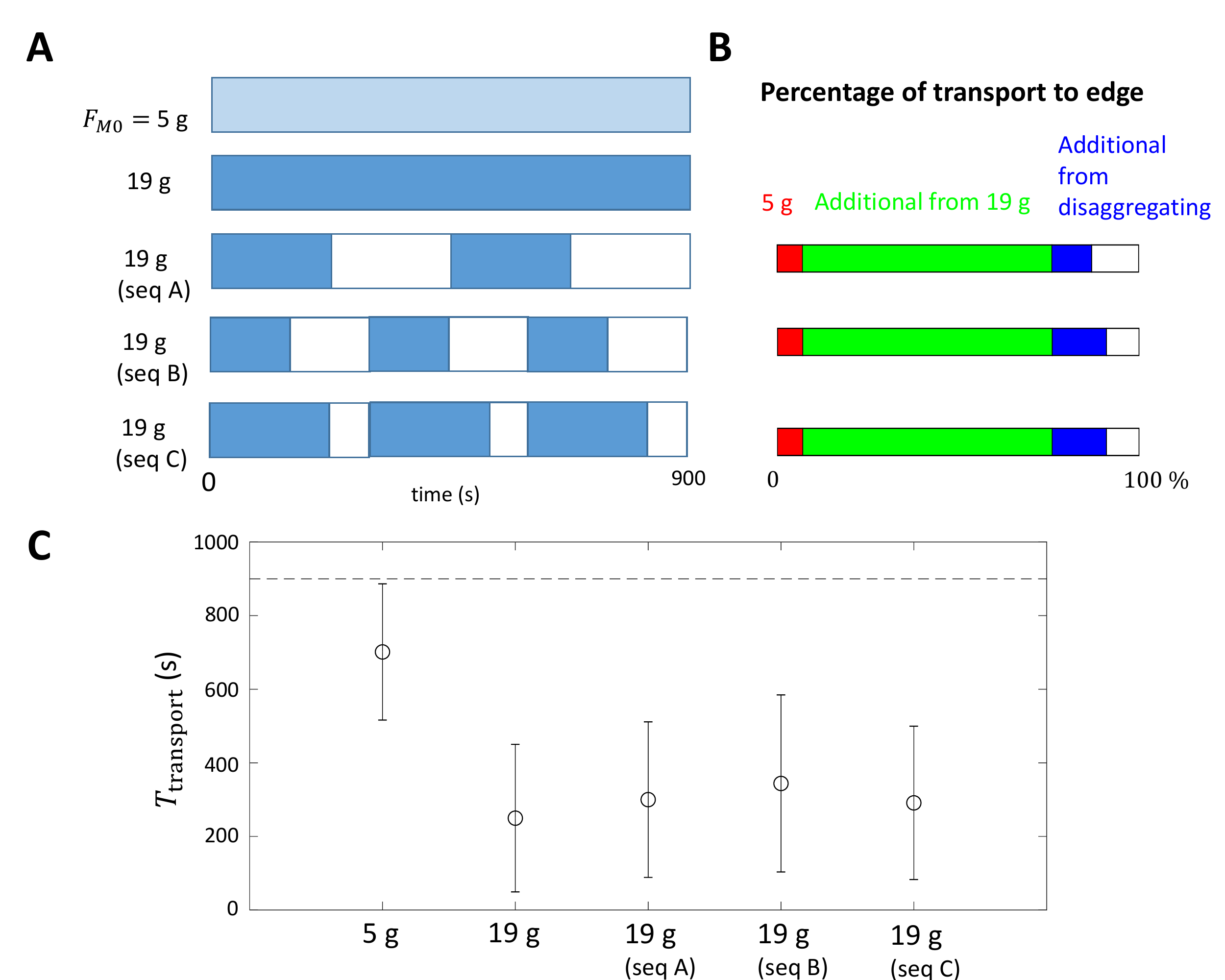}
    \caption{\textbf{Transport enhanced by disaggregating.}
    (\textbf{A}) Different patterns for magnetic strength over time used in simulations.
    The last three interleave periods of strong magnetic attraction with periods without any magnetic attraction (disaggregating).
    (\textbf{B}) The percentage of simulation runs achieved by weakly attractive collectives (red), strongly attractive collectives (green), and strongly attractive collectives with disaggregating (blue).
    The statistics use 100 simulations for each scenario.
    An animation of simulations using disaggregting sequence C can be found in Movie S7.
    (\textbf{C}) Comparison of transport time to the boundary among collectives that are weakly attractive, strongly attractive, and strongly attractive with disaggregating.}
    \label{fig:shuffling}
\end{figure}

\newpage

\subsection*{S7. From SOPS to the fixed-magnetization Ising model}

In this section, we prove that the SOPS algorithm can be mapped to an fixed-magnetization Ising model with coupling strength $J = \frac{1}{2\beta}\log{\lambda}$.
For a given SOPS configuration of particles in a bounded region of the lattice, construct a corresponding Ising lattice gas where the spin $\sigma$ of an occupied (resp., unoccupied) node in the SOPS is $+1$ (resp., $-1$) in the gas.
The SOPS algorithm has a transition probabilities $P_\text{SOPS} = \lambda^{-\Delta H_\text{SOPS}}$, where $H_\text{SOPS} = -n_\text{move}$ is its Hamiltonian and $n_\text{move}$ is the number of neighbors that the moving particle is leaving.
The Ising model has transition probabilities $P_\text{Ising} = \exp(-\beta\Delta H_\text{Ising})$, where $H_\text{Ising} = -\frac{1}{2}J\sum_{(i,j)}\sigma_i\sigma_j$ is its Hamiltonian and $\sigma_i \in \{-1, +1\}$ is the spin of site $i$.

\begin{lemma} \label{lem:sops2ising}
    Let $H_\text{SOPS}$ be the Hamiltonian of the SOPS algorithm, $H_\text{Ising}$ be the Hamiltonian of a fixed-magnetization Ising model, and $J$ be the coupling strength of the Ising model.
    Then for any particle move in the SOPS algorithm and the corresponding spin updates in the Ising model, we have $\Delta H_\text{Ising} = 2J\Delta H_\text{SOPS}$.
\end{lemma}
\begin{proof}
    Consider a particle moving from node $i$ to node $j$ in the SOPS algorithm and let $n_i$ (resp., $n_j$) be the number of neighbors the particle has at node $i$ (resp., node $j$).
    It is easy to see that $\Delta H_\text{SOPS} = n_j - n_i$ for this move.
    To calculate $\Delta H_\text{Ising}$, observe that the corresponding spin changes in the Ising model are $\sigma_i: +1 \to -1$ and $\sigma_j: -1 \to +1$.
    Let $z$ be the coordination number (i.e., degree) of the lattice.
    Consider all sites $k$ adjacent to $i$ and $j$; we have four cases:
    \begin{enumerate}
        \item $k$ is occupied and adjacent to $i$, so $\sigma_i\sigma_k: +1 \to -1$.
        There are $n_i$ such sites.
        \item $k$ is occupied and adjacent to $j$, so $\sigma_j\sigma_k: -1 \to +1$.
        There are $n_j$ such sites.
        \item $k \neq j$ is unoccupied and adjacent to $i$, so $\sigma_i\sigma_k: -1 \to +1$.
        There are $z - n_i - 1$ such sites.
        \item $k \neq i$ is unoccupied and adjacent to $j$, so $\sigma_j\sigma_k: +1 \to -1$.
        There are $z - n_j - 1$ such sites.
    \end{enumerate}
    To calculate $\Delta H_\text{Ising}$, we simply sum the spin changes of these cases; all other spins remain the same and thus cancel in the difference.
    We have:
    \begin{align*}
        \Delta H_\text{Ising} &= -\frac{1}{2}J((n_i - (z - n_i - 1) - n_j + (z - n_j - 1)) - (-n_i + (z - n_i - 1) + n_j - (z - n_j - 1)))\\
        &= -\frac{1}{2}J((2n_i - 2n_j) - (-2n_i + 2n_j))\\
        &= 2J(n_j - n_i)\\
        &= 2J\Delta H_\text{SOPS}
    \end{align*}
    This proves the lemma.
\end{proof}

\figtext~\ref{fig:IsingSOPS} shows examples of particle moves and their corresponding changes to the Ising Hamiltonian to illustrate the relationship established by Lemma~\ref{lem:sops2ising}.
Lemma~\ref{lem:sops2ising} shows that it is exactly when $J = \frac{1}{2\beta}\log\lambda$ that we achieve $P_\text{SOPS} = P_\text{Ising}$, which completes the mapping from the SOPS algorithm to the fixed-magnetization Ising model.

\begin{figure}[th]
    \centering
    \includegraphics[width=0.8\textwidth]{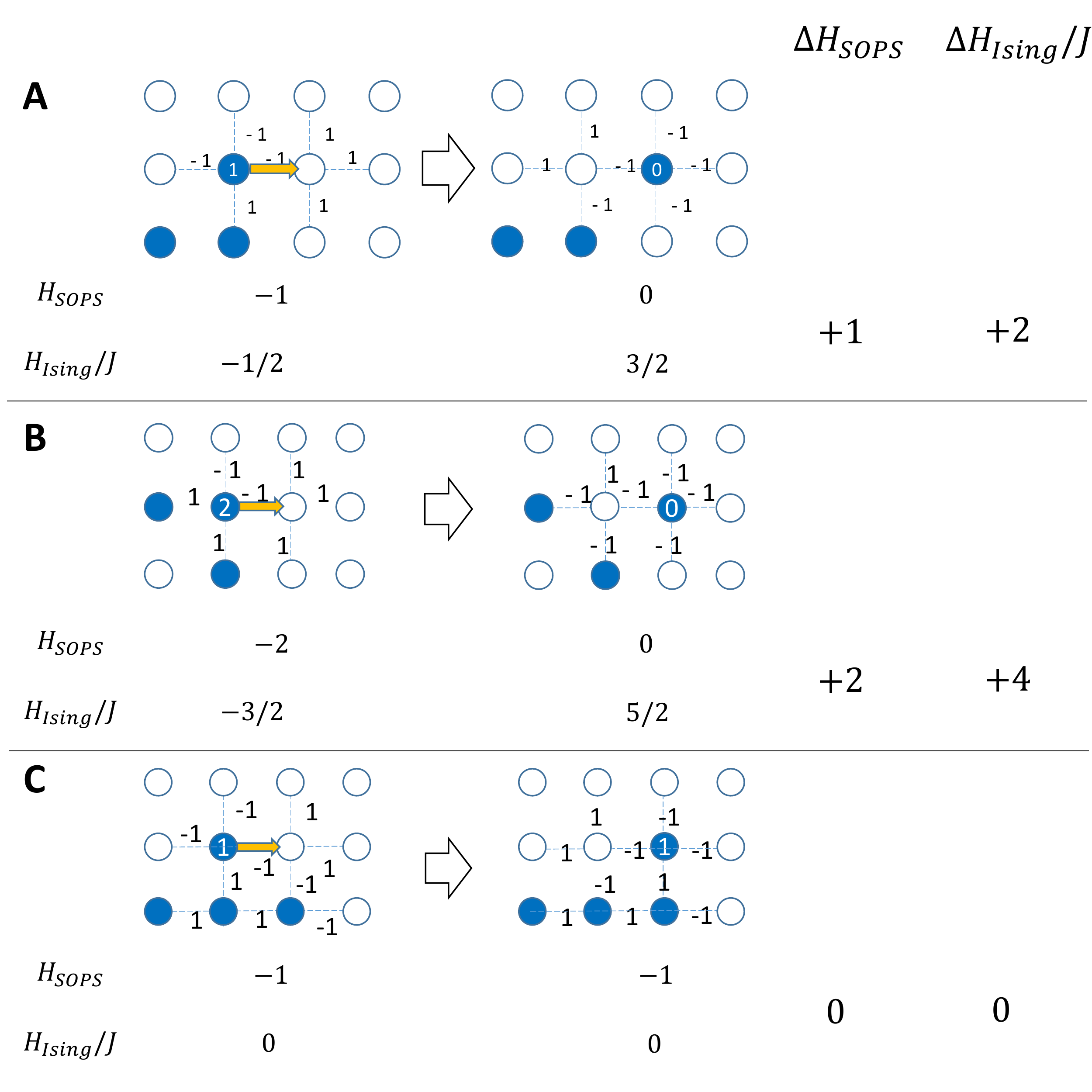}
    \caption{\textbf{Examples showing $\Delta H_\text{Ising} = 2J\Delta H_\text{SOPS}$.} 
    A particle moves to the right with a decrease of (\textbf{A}) 1, (\textbf{B}) 2, and (\textbf{C}) 0 neighbors, illustrated by the graphs showing the local configuration before and after the move.
    The numbers in black between sites $i$ and $j$ are the value of $\sigma_i\sigma_j$.
    The number of neighbors of the moving particle is shown in white.}
    \label{fig:IsingSOPS}
\end{figure}

\clearpage

\subsection*{S8. From the fixed-magnetization Ising model to the Cahn--Hilliard equation}

Penrose~\cite{penrose1991mean} has shown that the fixed-magnetization Ising model (i.e., Kawasaki dynamics) can be mapped to a Cahn--Hilliard equation by using a mean field treatment:
\[\frac{\partial u}{\partial t} = M\nabla^2\{f'(u) - \epsilon\nabla^2 u\}\]
where $M = \beta$, $\epsilon = J$, $f(u) = \beta^{-1}g(u) - \frac{1}{2}(z+1)\epsilon u^2$, $g'(u)=\text{arctanh}(u)$, and the prime denotes $d/du$.
Thus, we obtain:
\[\frac{\partial u}{\partial t} = \nabla^2\{\text{arctanh}(u) - (z + 1)\beta\epsilon u - \beta\epsilon\nabla^2 u\}\]
which is the standard Cahn--Hilliard equation:
\[\frac{\partial u}{\partial t} = \nabla^2(\Phi'(u) - \gamma\nabla^2 u)\]
with surface tension $\gamma = \beta J$ and $\Phi'(u) = \text{arctanh}(u) - (z + 1)\gamma u$.
Here, $z$ is the coordination number of the lattice: $4$ for square and $6$ for hexagonal.
Along with $J = \frac{1}{2\beta}\log\lambda$ proved in the previous section, we arrive at the connection between the SOPS algorithm and the Cahn--Hilliard equation with surface tension $\gamma = \frac{1}{2}\log\lambda$.

When $\lambda = 1$, the Cahn--Hilliard equation has no surface tension as $\gamma = \frac{1}{2}\log 1 = 0$, $\Phi'(u)$ has only one zero, and $\Phi$ has only one minimum.
As $\lambda$ increases, the surface energy $\gamma$ increases as well.
When $\gamma > 1/(z+1)$, $\Phi'$ has three zeros and $\Phi$ has double wells (\figtext~\ref{fig:bifurcation}), yielding a critical $\lambda_c = e^{2/7} \approx 1.33$ in the hexagonal lattice and $\lambda_c = e^{2/5} \approx 1.49$ in the square lattice.
The critical value $\lambda_c \approx 1.33$ for the hexagonal lattice lies within the $\lambda_c \in (1.02, 5.66)$ range predicted by the SOPS theory and exhibited by the BOBbot experiments.

\begin{figure}[th]
    \centering
    \includegraphics[width=0.7\textwidth]{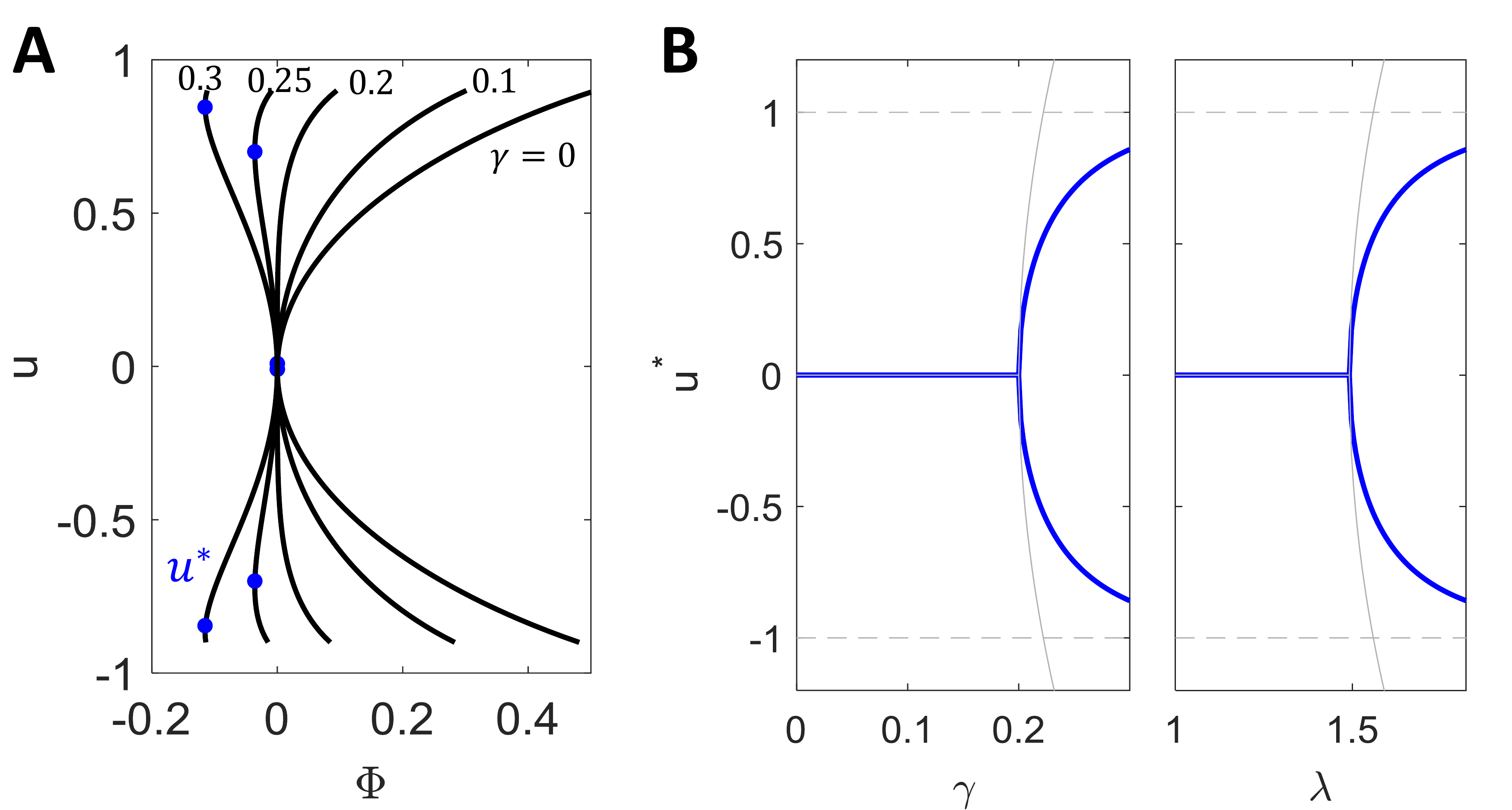}
    \caption{\textbf{Critical surface tension $\gamma$ and bias parameter $\lambda$.} 
    (\textbf{A}) Homogeneous free energy $\Phi$ for different surface tensions $\gamma$.
    (\textbf{B}) Position of the well for different surface tensions $\gamma$ and its corresponding bias parameter $\lambda$.
    The solid gray line shows the value $u^* = \pm\sqrt{15(\gamma-1/5)}$ from the Taylor expansion of $\Phi$ up until $\mathcal{O}(u^3)$.}
    \label{fig:bifurcation}
\end{figure}

When the surface tension is above the critical point, the characteristic length $\ell$ grows as $t^{1/3}$ \cite{toral1995large-sm}. Figure S15 shows how $\ell$ grows with time when the free energy mapped above when $\lambda$ is below and above the critical point for square lattice. To deal with the singular behavior of $\text{arctanh}$, linear extension around $\pm 1$ is used \cite{chen2019positivity-sm}. $\ell$ uses the first zero of the spatial correlation function $G(r)$, which is the Fourier transform of the structure factor \cite{toral1995large-sm}. Given the area scales as $N_{MC}\propto P_{MC}^{1/0.66} \approx P_{MC}^{3/2} = \ell ^{3/2}$ for the aggregated case in $N \sim 100 \ll \infty$, cluster size grows with time as $t^{1/2}$.

\begin{figure}[th]
    \centering
    \includegraphics[width=0.8\textwidth]{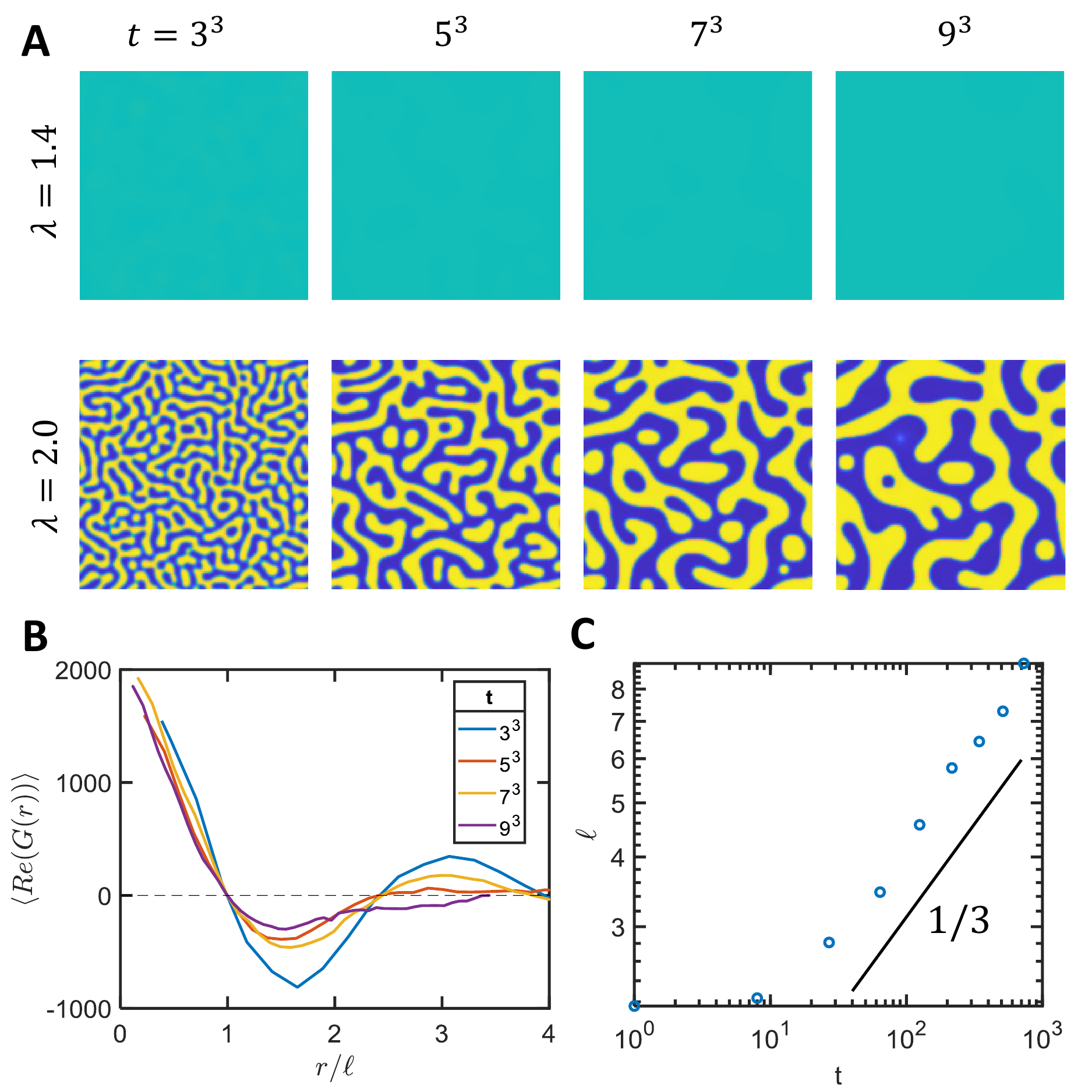}
    \caption{\textbf{Pattern formation below and above critical $\lambda$.} 
    (\textbf{A}) Simulation of the Cahn--Hilliard equation for $\lambda = 1.4$ below the critical point and $\lambda = 2.0$ above the critical point.
    Both simulations use a $128 \times 128$ grid and start with a uniform distribution of $u \sim \mathcal{U}[-0.1,0.1]$.
    (\textbf{B}) Spatial correlation function normalized by the correlation length $\ell$.
    (\textbf{C}) The increase of $\ell$ with time shows a power law with exponent $1/3$.}
    \label{fig:CHsim}
\end{figure}

\clearpage

\subsection*{S9. Approaching $P_{MC}\propto N_{MC}^{1/2}$ with thermodynamic limit}

The SOPS theory predicts that strongly attractive ensembles should produce aggregates that are both large and compact, namely, that the perimeter of the largest connected component $P_{MC}$ should scale with its size $N_{MC}$ with a $1/2$ power.
However, data from BOBbot experiments follow a slightly larger $0.66 \pm 0.07$ exponent (\figtext~\ref{fig:thermoDynLim}, blue).
This discrepancy is due in part to boundary and finite-size effects, so in simulation we investigated periodic boundary conditions (\figtext~\ref{fig:thermoDynLim}, red).
These simulations have a $0.59 \pm 0.18$ exponent that better aligns with the predictions of the SOPS theory.

\begin{figure}[th]
    \centering
    \includegraphics[width=0.7\textwidth]{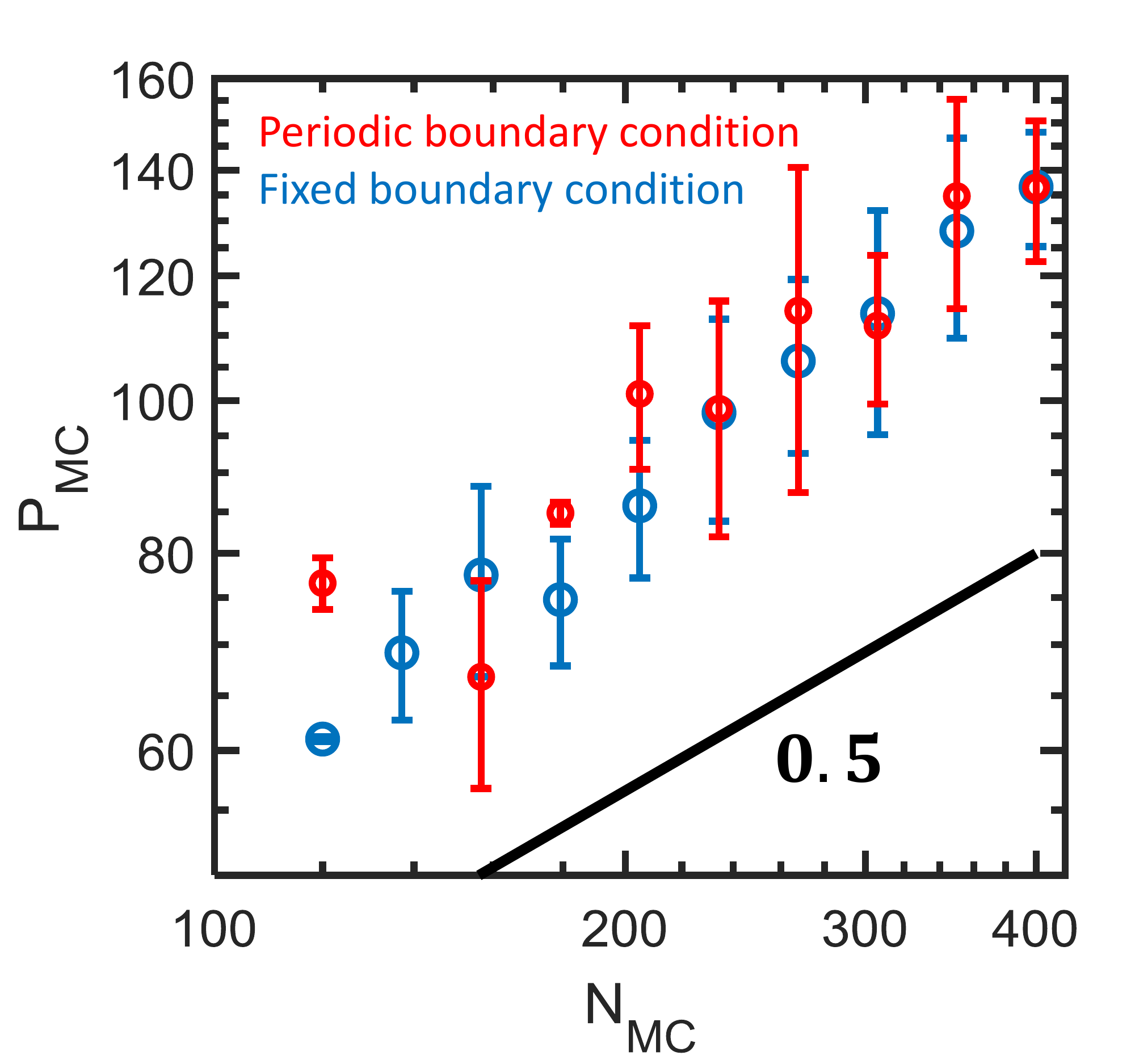}
    \caption{\textbf{Approaching $P_{MC} \propto N_{MC}^{1/2}$ with periodic boundary conditions.}
    Scaling between the largest component's size $N_{MC}$ and perimeter $P_{MC}$ in number of BOBbots for simulated systems of 100--400 BOBbots with $F_{M0} = 19$ g using fixed boundary conditions (blue) and periodic boundary conditions (red).
    The fixed boundary conditions achieve a scaling power of $0.66 \pm 0.07$ while periodic boundary conditions achieve a scaling power of $0.59 \pm 0.18$.}
    \label{fig:thermoDynLim}
\end{figure}

\end{document}